\documentclass[12pt,a4paper]{article}
\usepackage[margin=1in]{geometry}
\usepackage{amsmath,amssymb,amsthm,mathtools}
\newcommand{\Cov}{\operatorname{Cov}}

\usepackage{booktabs,longtable}
\usepackage{enumitem}
\usepackage{caption,subcaption}
\usepackage{graphicx}
\usepackage{tikz}
\usetikzlibrary{arrows.meta,positioning}
\usepackage{authblk}
\usepackage{setspace}
\usepackage{parskip}
\usepackage{hyperref}
\usepackage{booktabs}
\usepackage{multirow}
\usepackage{siunitx}
\usepackage{float}
\hypersetup{colorlinks=true, linkcolor=blue, citecolor=blue, urlcolor=blue}

\usepackage[numbers]{natbib}

\newtheorem{theorem}{Theorem}[section]

\newtheorem{proposition}[theorem]{Proposition}
\newtheorem{assumption}{Assumption}[section]
\theoremstyle{definition}

\theoremstyle{remark}
\newtheorem*{remark}{Remark}

\title{Long-Run Sovereign Debt Composition:
An Analytic Ergodic Framework with Explicit Maturity Structure}
\author{Christopher Cameron
\thanks{The author is with the U.S. Department of the Treasury, Office of Debt Management. The 
analysis and 
conclusions set forth in this paper are those of the author, and do not necessarily reflect those of or 
indicate concurrence by other members of the Treasury staff, Treasury's senior officials, the Treasury Department, or the United 
States government.}
\\ \textit{christopher.cameron@treasury.gov}
}
\date{February 12, 2026}

\begin{document}
\maketitle
 
\begin{abstract}
This paper describes a discrete-time model of regularly-issued sovereign debt dynamics under a deficit-driven nominal 
debt growth regime that explicitly accounts for granular maturity. New 
issuance follows fixed allocations across a finite maturity ladder, and the government budget constraint determines total borrowing endogenously. In the 
deterministic baseline, we identify a sustainability condition for convergence to a steady-state and
derive closed-form steady portfolio shares, as well as key metrics for steady cost and risk (proxied as one-period
rollover ratio). Extending the model to a 
stochastic recurrence equation (SRE) driven by interest rates and (normalized) deficits that 
are stationary and mean-reverting, and
using a future-cashflow state representation of debt, we 
identify an analogous condition for ergodic convergence 
to a unique invariant distribution. This implies that metrics calculated by
Monte Carlo debt simulations driven by factors with these properties will recover the ergodic means of the underlying
system, independently
of initial conditions, provided the simulation horizon is sufficiently long. Analytical formulae for expectations of 
certain key metrics under this invariant
distribution are derived, and agreement with simulation is observed. We find that the 
introduction of stochastic interest-rate/deficit correlation into the framework 
leads to intuitive correction terms to their deterministic-baseline counterparts.
\end{abstract}

\section{Introduction}
\label{sec:intro}

The management of sovereign debt portfolios is a critical aspect of fiscal policy, balancing the trade-offs between minimizing interest cost and mitigating risks associated with refinancing, operational overhead, and interest rate fluctuations. Traditional analyses of debt sustainability (\citep{escolano2010}, \citep{cottarelli2012fiscal}) often either aggregate debt stock without distinguishing by maturity, overlooking the heterogeneous impacts of maturity structures on rollover needs and interest payments, or use stylized decompositions into short and long maturity buckets that elide the full maturity distribution. This 
aggregation may omit details of interest to debt managers seeking to understand the granular interest-cost and risk impact of maturity selection and
issuance-size adjustments.

Popular modeling frameworks (e.g. \citep{bolder2003}, \citep{bolder2011}, \citep{consiglio2012}, \citep{belton2018}) deploy 
numerical forward debt simulations in a stochastic framework, from which averaged 
metrics are calculated empirically, and tradeoffs between debt management goals  explored or optimized. Such models
allow for the ability to model issuance allocations and time-periods of any desired granularity, as well as to incorporate an arbitrarily detailed set of 
assumed macroeconomic factors and relationships as drivers. Because of the number of variables often involved, 
and the large dimensionality of the relationships being simulated, the convenience and flexibility of such models can 
come at the expense of clarity and intuition regarding the results.

In this paper, we describe an approach, first introduced in \citep{cameron2018} and \citep{landoni2019}, to analytically bridging the gap between aggregated theory on the one hand, and 
detailed but conceptually complex empirics on the other, using a disaggregated, discrete-time model of 
fixed-rate debt dynamics that explicitly accounts for the full maturity ladder. Our framework tracks new and outstanding debt by vintage and tenor, with new issuances allocated according to fixed 
proportions across maturities. Minimal issuance under the standard 
government budget constraint endogenously fixes total new borrowing to exactly cover primary deficits, interest 
payments, and maturing principal. The resulting dynamics, cast as recursion equations, allow us to analyze both 
deterministic and simple stochastic frameworks, providing insights into steady-state behaviors and risk profiles of debt portfolios under regular-issuance policies.

In the deterministic baseline, we find closed-form expressions for steady-state portfolio shares in an assumed deficit-driven regime of exponentially growing 
(nominal) primary deficits, equivalent to positing a sustainability condition analogous to $r<g$ (\citep{blanchard2019}). These 
shares depend solely on issuance weights and the deficit growth rate, and are independent of the coupon structure (provided
deficit-growth prevails). The steady-state portfolio also has closed-form formulas for yearly interest cost and risk (proxied as the 
one-period rollover fraction) that depend solely on the new-issuance allocation, steady coupons, 
and the deficit growth rate.

Extending to stochastic settings, we incorporate mean-reverting interest rates and primary deficits, modeled in this case 
as AR(1) processes, which can be correlated to capture countercyclical policy responses. By invoking Foster-Lyapunov drift conditions on 
the resulting random-coefficient stochastic recurrence equation (SRE)
(\citep{meyntweedie1993}, \citep{hairer2011yet}, \citep{bougerol1993products}), we 
establish, again under a sustainability condition that implies deficit-driven growth, ergodic 
convergence to a unique stationary distribution. Here, we use a state representation of ongoing debt issuance based on tracking all future cashflows, which preserves and exposes the underlying
linearity of the stochastic recurrence equation driving dynamics.

Ergodic convergence ensures that long-run pathwise time averages, as well as stochastic averages across multiple realizations, of the debt process driven 
in this way
align with its theoretical long-term invariant distribution. It follows that for numerical debt models with the same properties, observable metrics computed from 
long-horizon simulations should largely recover those of the invariant distribution, provided the simulation horizon is long enough for 
mean-reversion to dominate the effect of initial conditions.

Our results have direct implications for debt-management offices (DMOs) seeking to understand the long-term cost-risk tradeoff both theoretically, and for interpreting the empirical results of numerical debt simulation,
when contemplating issuance allocation and conveying the fundamental tradeoffs and relationships between maturity selection. The model also 
lays groundwork for incorporating floating-rate debt and other extensions, while also illuminating the difficulty in incorporating such 
economically realistic factors as endogenous
interest rates.

The literature on sovereign debt management is extensive but often focuses on continuous-time models or aggregate dynamics.
The discrete-time, full-maturity-ladder, and future-cashflow-representation nature of the
approach described here
complements these by providing tractable closed-forms that allow for incorporation of stochastic interest rate and deficit drivers, and ergodic 
analyses, while maintaining the granularity of the maturity profile, bridging theoretical rigor with practical application and computability.

The paper proceeds as follows: Section~\ref{sec:det} introduces the mathematical notation 
and presents the deterministic baseline model introduced in an earlier piece, derives the conditions for asymptotic convergence to a unique 
steady-state, and presents key steady-state
metrics. Section~\ref{sec:stoch} extends the model framework to incorporate stochastic, mean-reverting,
correlated interest-rates and derives the appropriate generalization for ergodic convergence to an invariant distribution. Formulae for 
the key metrics
from the deterministic baseline obtain appropriate correlation adjustments. Section~\ref{sec:numerics}
depicts numerical examples of the SRE model for a simple baseline parameter set, comparing and showing 
fidelity of the results of numerical Monte Carlo simulation to 
the theoretical invariant expectations of key metrics. Section~\ref{sec:concl} concludes.

\section{Deterministic Baseline}
\label{sec:det}

This section establishes the baseline deterministic model, defining the state variables, issuance rules, payment structures, and budget constraints. We derive conditions for sustainability under exponential deficit growth and obtain closed-form steady-state portfolio shares.

We largely follow the framework and notation introduced in \citep{cameron2018}. 

\subsection{Model Setup}

Time is modeled in discrete periods $t \in \mathbb{N}_0 = \{0, 1, 2, \dots\}$. The economy features a sovereign borrower issuing debt across a finite set of integer 
maturities $j = 1, 2, \dots, M$, where $M$ is the maximum maturity in periods. All quantities are nominal where applicable, unless otherwise stated.

\begin{assumption}[Debt State]
\label{ass:state}
The outstanding debt at the beginning of period $t$ is represented by the vector
\[
{Q}_t = (Q_{t,1}, Q_{t,2}, \dots, Q_{t,M})^\top \in \mathbb{R}^M,
\]
where $Q_{t,j}$ denotes the face amount of debt with  $j$ periods remaining to maturity at period $t$.
\end{assumption}

Debt evolves as existing obligations roll down the maturity ladder, and new issuances are added.

\begin{assumption}[Issuance Policy]
\label{ass:issuance}
Total new issuance $N_t$ in period $t$ is allocated across maturities according to fixed proportions:
\[
N_{t,j} = f_j N_t, \quad j = 1, \dots, M,
\]
where ${f} = (f_1, \dots, f_M)^\top$ does not depend on $t$, and satisfies $f_j \ge 0$ and $\sum_{j=1}^M f_j = 1$.
\end{assumption}

This fixed-proportion rule is a stylized debt-management issuance strategy. The steady linkage of assumed issuance allocations to the flow of 
new funding needs (as opposed to the stock of existing debt) represents a DMO that
follows a stable issuance practice (such as the US, with its `regular \& predictable' framework) rather than one that is more opportunistic and volatile
in its debt flows and debt-management practices, or which actively attempts to modify the outstanding portfolio to guide it toward target metrics.

\begin{assumption}[Coupon Payments]
\label{ass:coupon}
A bond issued at time $t$ with original maturity $j$ pays a fixed coupon $r_j \ge 0$ on its face value $N_{t,j}$  in periods $t+1$ through $t + j$, totaling $j$ payments.
\end{assumption}

Coupons in this model are assumed to be paid periodwise for simplicity (i.e. if the period is yearly, then
semiannual or other coupon schedules are considered to be rolled up by period). The coupon rate 
itself $r_j$ is assumed constant for each maturity but may vary across $j$ to reflect a term structure. The `yield curve' of rates 
$\{r_j\}_{j=1}^M$ can be thought of as the long-term steady yield curve assumption. Typically, this would be assumed to be upward-sloping (though
this is not essential to the basic framework); more will
be said about the importance and role of the assumed curve shape later.

Total interest payments in period $t$ are the sum of coupons on all outstanding debt issued 
in prior periods:
\begin{equation}
\label{eq:It_det}
I_t = \sum_{j=1}^M \sum_{s=t-j}^{t-1} r_j N_{s,j} 
 = \sum_{j=1}^M r_j f_j \left( \sum_{s = 1}^j N_{t-s}\right) =
 \sum_{j=1}^M \left( \sum_{k=j}^M r_k f_k \right) N_{t-j}.
\end{equation}

Maturing principal in period $t$ is the debt issued in prior periods that reaches zero remaining maturity at period $t$:
\begin{equation}
\label{eq:Mt_det}
M_t = Q_{t-1,1} = \sum_{j=1}^M N_{t-j,j} = \sum_{j=1}^M f_j N_{t-j}.
\end{equation}

\begin{assumption}[Government Budget Identity]
\label{ass:budget}
The government finances the primary deficit $D_t$ (representing any new financing-needs), 
interest $I_t$, and maturing principal $M_t$ through new issuance $N_t$:
\[
N_t = D_t + I_t + M_t.
\]
\end{assumption}

This identity closes the model endogenously, with $N_t$ adjusting to meet fiscal needs. Since the maturing term $M_t$ and interest cost $I_t$  depend linearly
on quantities $N_s, s < t$, this is a nonhomogenous linear recursion equation for $N_t$ driven
by the exogenous forcing-term $D_t$. 
\begin{equation} \label{eq:budget}
N_{t} = D_t + I_t + M_t = \sum_{j=1}^M \left(   ( \sum_{k=j}^M r_k f_k                                )   + f_j    \right) N_{t-j} + D_t
\end{equation}

\begin{assumption}[State Evolution]
\label{ass:state_evol}
The debt state updates as:
\[
Q_{t,M} = N_{t,M} = f_M N_t, \quad Q_{t,j} = Q_{t-1,j+1} + N_{t,j}  = Q_{t-1,j+1} + f_j N_{t}, \quad j = 1, \dots, M-1.
\]

\end{assumption}
In matrix form:
\[
{Q}_{t} = S {Q}_{t-1} + N_{t} {f}  = S {Q}_{t-1} +  {f} (D_t + I_t + M_t)
\]
where $S$ is the $M \times M$ shift matrix with $S_{j,j+1} = 1$ for $j=1,\dots,M-1$, and zeros elsewhere.

\subsection{Exponential Deficits and Growth-Regime}

To analyze long-run behavior, we impose structure on the primary deficit $D_t$. Our assumption here is the most general one that can still plausibly
admit of a self-similar portfolio and thus a steady-state in normalized quantities.

\begin{assumption}[Exponential Deficit Growth]
\label{ass:exp_def}
Primary (nominal) deficits grow exponentially:
\[
D_t = D_0 \gamma^t, \quad \gamma := 1+g > 1 \quad (g>0), \quad D_0 > 0.
\]
\end{assumption}

In 
classical debt models it may be assumed, or imposed as a sustainability constraint, that normalized quantities such as debt/GDP and/or primary
deficit/GDP converge to a sustainable long-term fixed quantity. We may connect with that literature by positing $g := \gamma-1$ as the assumed 
long-term steady GDP growth-rate. Our assumption about $D_t$ then merely represents deficit-growth that maintains stable debt/GDP. While 
this identification of $\gamma$ with the GDP-growth factor is not essential to the basic mathematical framework used here, we note this simply to 
observe that such an assumption about the trend of $D_t$ is (if only implicitly) not uncommon.

\subsection{Backward Recursion and Steady-State}

We now analyze the baseline recursion for $N_t$. First normalize all variables by the growth trend: $\tilde{N}_t = N_t / \gamma^t$, $\tilde{D}_t = D_t / \gamma^t = D_0$, $\tilde{{Q}}_t = {Q}_t / \gamma^t$.

Substituting into the budget identity~\eqref{eq:budget} and normalizing yields a recurrence for the normalized new-issuance $\tilde{N}_t$:
\[
\tilde{N}_t = D_0 + \sum_{j=1}^M f_j \gamma^{-j} \tilde{N}_{t-j} + \sum_{j=1}^M r_j f_j \sum_{s=1}^{j} \gamma^{-s} \tilde{N}_{t-s}.
\]

We are interested in whether the normalized system approaches a long-term steady limit, i.e. $\tilde{N}_t \to \tilde{N}_{\infty}$.  Such a limit must satisfy

\begin{align*}
\tilde{N}_{\infty} &= D_0 + \left( \sum_{j=1}^M f_j \gamma^{-j} + \sum_{j=1}^M r_j f_j \sum_{s=1}^{j} \gamma^{-s}  \right) \tilde{N}_{\infty} \\
  &= D_0 + \left(  \sum_{j=1}^M \left( \sum_{k=j}^M r_k f_k + f_j \right) \gamma^{-j}  \right) \tilde{N}_{\infty}
\end{align*}

Define the feedback function:
\begin{equation} \label{eq:feedback}
\Phi(\gamma, {r}, {f}) = \sum_{j=1}^M \left( \sum_{k=j}^M r_k f_k + f_j \right) \gamma^{-j}    = \sum_{j=1}^M f_j \gamma^{-j} + \sum_{j=1}^M r_j f_j \gamma^{-1} \frac{1 - \gamma^{-j}}{1 - \gamma^{-1}}.
\end{equation}

\begin{proposition}[Deficit-Driven Regime]
\label{prop:sustain}
The system defined by \ref{ass:state}--\ref{ass:exp_def} is 
deficit-driven with growth rate $\gamma$ if $\Phi(\gamma, {r}, {f}) < 1$, where $\Phi(\cdot)$ is given by~\eqref{eq:feedback}. In this case, the normalized issuance 
converges to $\tilde{N}_\infty = D_0 /(1- \Phi(\gamma, {r}, {f})) > 0$. 
In other words, for large $t$, new 
debt asymptotically grows at a rate $N_t \sim \tilde{N}_{\infty} \gamma^t$ 
that depends only on $\gamma$, the growth of deficits, not on interest rates, and
independently of initial conditions.
\end{proposition}

\begin{proof}
The recurrence equation for $\tilde{N}_t$ is a linear homogeneous operator on prior values $\{N_s\}_{s<t}$ plus a constant term $D_0$. Asymptotic convergence under the condition on $\Phi$ (here and below we may suppress dependence on $\gamma, 
{r}, {f}$) can be shown directly by iterating the equation for
aggregate debt (see \citep{cameron2018}); $\Phi$ characterizes the largest root of the associated characteristic polynomial.

Here, to facilitate the discussion below, we will
make the linear recursion explicit by defining an augmented state:
\[
{X}_t = (\tilde{N}_t, \tilde{N}_{t-1}, \dots, \tilde{N}_{t-(M-1)})^\top \in \mathbb{R}^{M},
\]

The evolution of ${X}_t$ thus has the form:
\[
{X}_{t} = B {X}_{t-1} + {d}
\]
where  

\begin{itemize}
\item $d = D_0 {e}_1$, with $e_1 = (1,0,0,\dots)^T$, and 
\item the companion matrix $B$ is constant (in $t$), depending only on $\{r_j\}$, $\gamma$ and ${f}$:
\[
B =S + e_1 b^T,
\]
where $S = (\delta_{j+1,j})$ is an $M \times M$ shift operator containing only 
entries $1$ in the first subdiagonal, 
and $b_j =  (\sum_{k=j}^M r_k f_k + f_j) \gamma^{-j}$. (Note $B$ is a Leslie matrix, as arises
in population ecology; many results from that field will obtain in a laddered-issuance debt model.)
\end{itemize}

The condition $\Phi < 1$ means that the first-row sum of $B$, i.e. $\sum b_j$, is $<1$ (while the other row sums are 
all $1$ because $S$ is just a shift operator). $B$ is nonnegative and irreducible, so
by an application of the Perron-Frobenius theorem, the $\Phi<1$ condition means that
the spectral radius of $B$ satisfies $\rho(B)<1$. This ensures 
that the homogeneous solution to the recursion decays from any initial condition, and 
that the particular solution dominates as steady attractor. As we have seen, the limit of $X_{t,1} = \tilde{N}_t$
can be written $\tilde{N}_{\infty} = D_0 / (1-\Phi)$. 
\end{proof}

We will see later that condition~\ref{prop:sustain} is a disaggregated analogue of "$r<g$" in traditional debt-sustainability 
literature (i.e. if we identify $g=\gamma-1$ with the GDP growth rate).

\subsection{Steady-State Portfolio Shares}

The normalized state $\tilde{Q}_t$ evolves as:
\[
\tilde{Q}_{t,M} = f_M \tilde{N}_{t}, \quad \tilde{Q}_{t,j} = \gamma^{-1} \tilde{Q}_{t-1,j+1} + f_j \tilde{N}_{t}, \quad j < M.
\]

(Note that $\tilde{Q}_{t-1,1}$ is the (normalized) outstanding amount with one period remaining and so it matures at step $t$.)

In steady state, $\tilde{Q}_{t,j} \to \tilde{q}_j$, yielding:
\[
\tilde{q}_M = f_M \tilde{N}_\infty, \quad \tilde{q}_j = \gamma^{-1} \tilde{q}_{j+1} + f_j \tilde{N}_\infty,
\]
where we recall $\tilde{N}_\infty = D_0 / (1-\Phi)$.

Solving backwards:
\begin{equation}\label{eq:qlevels}
\tilde{q}_j = \tilde{N}_\infty \sum_{k=j}^M \gamma^{j-k} f_k.
\end{equation}

The aggregate normalized debt in steady-state is:
\begin{equation} \label{eq:edebtlevel}
\tilde{q} = \sum_{j=1}^M \tilde{q}_j = \tilde{N}_\infty \sum_{k=1}^M f_k \frac{1 - \gamma^{-k}}{1 - \gamma^{-1}}.
\end{equation}

Thus, the steady-state portfolio shares are:
\begin{equation}
\label{eq:shares_det}
\theta_j = \frac{\tilde{q}_j}{\tilde{q}} = \frac{\sum_{k=j}^M f_k \gamma^{j-k}}{\sum_{k=1}^M f_k \frac{1 - \gamma^{-k}}{1 - \gamma^{-1}}}.
\end{equation}

\begin{proposition}[Convergence to Steady-State Shares]
\label{thm:shares}
For the system defined by \ref{ass:state}--\ref{ass:exp_def}, under the assumption 
$\Phi(\gamma,{r},{f}) < 1$, portfolio shares $Q_{j,t} / Q_t$ converge to
the values $\theta_j$ defined in \eqref{eq:shares_det}.
\end{proposition}

\begin{proof}
The normalized $\tilde{{Q}}_t$ converges to ${q}$ by the contraction property of the backward recursion with factor $\gamma^{-1} < 1$. The 
ratios $Q_{j,t} / Q_t = \tilde{Q}_{j,t} / \sum_k \tilde{Q}_{k,t}$ then converge to $\theta_j$, independent of initial conditions.
\end{proof}

\begin{remark}
The shares \eqref{eq:shares_det} depend only on ${f}$ and $\gamma$, not on $r$. Coupons influence the condition for deficit-driven growth,
and potentially sustainability, via $\Phi(\gamma,{r},{f})<1$, but do not affect the long-run portfolio composition in deficit-driven regimes.
\end{remark}

This deterministic baseline provides a benchmark for the stochastic extensions, for which
uncertainty introduces probabilistic risk dimensions. In the following sections
we observe some of its basic cost and risk properties.

\subsection{Steady Interest-Cost Ratio}
\label{WACsection}
In the deterministic steady-state, the next-period interest-cost, as a fraction of current-face outstanding, is just the steady weighted-average coupon. Because
rates $r_j$ are constant by original tenor, this can be written:
\[
WAC = \sum_{j=1}^M w_j r_j
\]
where the weights $w_j$ represent 
current-outstanding by \textit{original} maturity, rolling up the $j$ prior periods of new issuance at tenor $j$. 

Writing $w_{t,j}$ to represent
the fraction of outstanding debt at period $t$ with original-maturity $j$ in steady-state, we must have
\[
w_{t,j} \propto \sum_{s=t-j}^{t-1} f_j N_s = f_j \sum_{s=t-j}^{t-1} \gamma^s \tilde{N}_s
\sim \tilde{N}_{\infty} f_j \sum_{s=t-j}^{t-1} \gamma^s
 \]

After normalizing and cancellation we obtain explicit values for the portfolio weights that determine steady WAC:
\begin{equation} \label{eq:wac}
WAC = \sum_{j=1}^M w_j r_j, \; w_j = \frac{ f_j ( 1 - \gamma^{-j} ) }{ \sum_{k=1}^M f_k (1-\gamma^{-k})}
\end{equation}

\subsection{Steady Rollover Fraction}

The fraction of the deterministic baseline portfolio that rolls over each period is just $\theta_1$. Using $w_j$ as defined above, and substituting
into express \eqref{eq:shares_det} for $\theta_1$,  this can be expressed
\begin{equation} \label{eq:rr}
\theta_1 = \frac{q_1}{q} = \frac{\sum_{k=1}^M f_k \gamma^{1-k}}{\sum_{k=1}^M f_k \frac{1 - \gamma^{-k}}{1 - \gamma^{-1}}}
=  \frac{\sum_{k=1}^M f_k \gamma^{-k} }{ \sum_{k=1}^M f_k \frac{1 - \gamma^{-k}}{\gamma - 1}}
= \sum_{j=1}^M w_j \tau_j
\end{equation}
with $\tau_j = (\gamma-1)/(\gamma^j-1)$, and $w_j$ as defined for $WAC$,
\[
 w_j = \frac{ f_j ( 1 - \gamma^{-j} ) }{ \sum_{k=1}^M f_k (1-\gamma^{-k})}.
\]

(It will become apparent later why it is beneficial to express $\theta_1$ in terms of $w$.)

\subsection{Relationship to Sustainability Condition}

Using the preceding steady quantities,  algebraic manipulation shows
the feedback function $\Phi(\gamma,{r},{f})$ can 
be rewritten as simply $\Phi = (\theta_1+WAC)/(\theta_1+g)$. Hence condition~\ref{prop:sustain} for a deficit-driven regime with a unique steady-state
is equivalent to
\begin{equation} \label{rvg}
\Phi(\gamma,{r},{f}) =  \frac{ \theta_1 + WAC }{\theta_1 + g} < 1 \quad \iff \quad  WAC < g.
\end{equation}
This reveals the condition for being in a deficit-driven regime as the (intuitive if not 
tautological) statement that long-term deficits grow faster than steady-state compounded interest. As 
noted previously this is effectively the sustainability condition "$r < g$" that arises in debt literature, albeit to make the linkage sharp we would need to 
additionally identify $g$ with GDP growth as well (in our baseline framework it merely denotes deficit-growth).

In other words, if deficits grow at the same rate as GDP (no fiscal blowup), and the interest rates are such that the
sustainability condition $r<g$ is satisfied, 
where $r:=WAC$ is the steady, long-term average interest rate on the portfolio, then 
debt growth is deficit-driven, growing at the same asymptotic rate as GDP, and further, has a unique asymptotic steady-state from any initial condition.

Figure~\ref{fig:phi1} illustrates the relationship between $r, f$, and $g$ (or $\gamma$) vs. $\Phi$ for simple choices of new-issue allocations.
\noindent
\begin{figure}[H]
\centering
\includegraphics[scale=.4]{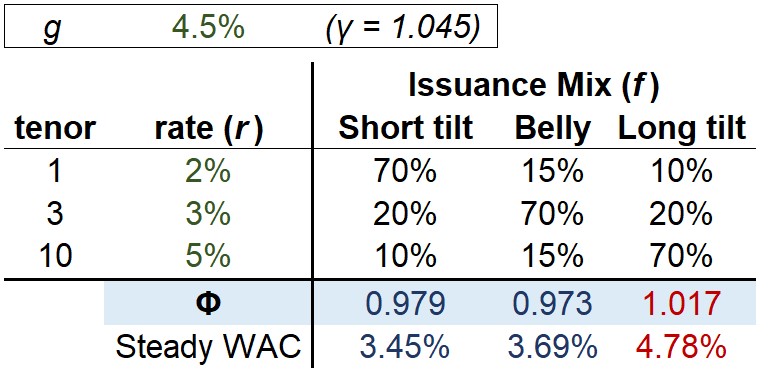}
\caption{\footnotesize Feedback function $\Phi$ and steady $WAC$ formula
calculated for $\gamma=1.045$, $r = (.02,.03,.05)^T$ under various issuance allocations. In
this example the long-tilted allocation leads to $\Phi > 1$ (the formula for $WAC$ leads to $WAC>g$), and so 
debt dynamics are interest-driven rather than deficit-driven. The other
two allocations satisfy $\Phi<1$ for this $\gamma$, thus attain steady-state under deficit-driven asymptotics.}
\label{fig:phi1}
\end{figure}

From the form~\eqref{rvg} of our deficit-driven (sustainability) condition we can also observe the intuitive result that, for given values of $WAC<g$, larger 
values of $\theta_1$ - that is, higher yearly rollover fractions - cause the portfolio to be closer to breaching its sustainability (deficit-driven steady-state) boundary, and 
convergence/pullback to the long-term trend would naturally be expected to be slower and more volatile.

\subsection{Optimization and Frontier}

In this section we describe how the preceding baseline metrics can be used to set up a simple if idealized portfolio selection problem.

Recall the expression for steady one-period rollover fraction,
\[
\theta_1 = \sum_{j=1}^M w_j \tau_j
\]
with $\tau_j = (\gamma-1)/(\gamma^j-1)$.

Note $\tau = \{ \tau_j \}$ is uniformly decreasing in $j$. Compare this 
form of $\theta_1$ with equation \eqref{eq:wac} for $WAC$:
\[
WAC = \sum_{j=1}^M w_j r_j
\]
Typically, $r$ used here will be upward-sloping. One then sees in the opposite slopes of $r$ and $\tau$ the 
classical cost-risk tradeoff.  Assuming cost is proxied by WAC and risk by $\theta_1$, the portfolio selection 
problem reduces to selecting $f$  to minimize $WAC$, subject to a constraint on $\theta_1$. 
Given the preceding formulas, and  
because
$f$ and $w$ differ only by a rescaling and renormalization (any $f$ can be uniquely recovered from its corresponding $w$ and
vice versa),
this reduces to solving for optimal accumulation weights $w$ to balance minimizing a weighted-sum of 
$r$ against keeping a weighted-sum of the downward-sloping $\tau$ less than a 
risk-tolerance. 

When framed this way, optimal issuance allocation from a steady-state viewpoint reduces to a simple linear programming 
problem. If $r$ is upward-sloping and not too concave, as would be typical of any plausible long-term yield curve assumption, then absent further constraint the optimal 
solution will inevitably place
maximum issuance weight on some intermediate or "belly" tenor
$j^*$ (or a pair of neighbor tenors). Namely, optimization will concentrate issuance on the smallest tenor (to minimize cost) such that 
concentrated-issuance at that ``sweet spot'' tenor, or a pair of neighboring tenors if the optimal tenor is not an integer, maintains $\theta_1$ at the prescribed level. Indeed, it 
can be shown under these conditions that the optimal tenor is just
\[
j^* =\frac{ \log( 1+ (\gamma-1)/R)}{\log \gamma} \approx 1/R,
\]
where $R$ is the risk tolerance for $\theta_1$, i.e. $\theta_1\le R$ is the constraint. (See \citep{cameron2018}, \citep{landoni2019}.)

If nontrivial constraints are applied (such as lower- and upper-bounds on the weights $f_j$) then 
some blended allocation within the feasible set, and best approximating this $j^*$, meaning issuance concentrated toward the ``belly'' rather than short
or long, will be the optimum. 

If, alternatively, $r$ does not have an upward-sloping structure then less can be said in general, as the simplex optimum would depend on the precise details of $r$. For example, if $r$ is ``barbelled'' (such that $r_1$ and $r_M$ are small while intermediate $r_j$ are larger), or highly concave, then one can find it optimal for issuance to follow a ``barbell'' pattern as well, with higher weight on $f_1$ and $f_M$ and lower or none at intermediate tenors. (Again 
see \citep{cameron2018}.) Such
an
assumption for $r$ would of course not be typical or parsimonious in debt modeling.

We add that absent further assumptions or established facts about the dynamics of the debt portfolio, the portfolio selection framework as described in this
section
is highly idealized, in effect answering the hypothetical 
question, if rates and (normalized) deficits were constant, which allocation $f$ would, if permanently deployed, lead to the optimal steady portfolio. While such a framework may seem overly simplified, 
we will find in the sequel that in a simple extension to a stochastic framework, to some extent many of the mathematical consequences of this baseline case can be
carried over and remain applicable 
to more realistic frameworks and models, with only modest adjustment.

\subsection{No-Growth Limit $\gamma \to 1^+$}

It is instructive to examine the asymptotic limits of these metrics and observations for the limiting case $\gamma \to 1^+$. 
We will show that the above steady metrics reduce to simple and intuitive accumulation quantities.

The steady portfolio fractions $\theta_j$ reduce, via $(1-\gamma^{-k})/(1-\gamma^{-1}) \to k$, to
\[
\theta_j \to \frac{\sum_{k=j}^M f_k}{\sum_{k=1}^M k f_k}
\]

Thus, the steady portfolio-share with tenor $j$ (under no deficit-growth) is just proportional to the fraction 
of new-issuance at/above tenor $j$. (The more realistic 
case $\gamma > 1$ applies a correction to this expression that preferentially weights more-recent issuance.)

For WAC, a similar reduction shows
\[
WAC \to \sum_j w_j r_j, \quad w_j = j f_j / \sum_k k f_k
\]
Note this weights each issuance fraction $f_j$ by how many periods the bucket remains in the portfolio, and then normalizes. Issuances of tenor $j$
from $j$ prior periods are present in the current portfolio, so $r_j$ must be counted $j$ times.

And for 1-period rollover, we can read off from above,
\[
\theta_1 \to \sum_k f_k / \sum_k k f_k = 1 / \sum_k k f_k.
\]
Note this is just the reciprocal of the new-issue weighted average maturity (NWAM). Unsurprisingly, shorter-/longer-tilted new-issuance leads to a steady
portfolio with higher/lower steady periodwise rollover fraction.

And as we have seen, there is an optimal ``belly'' tenor given risk tolerance $R$, which as $\gamma \to 1$ 
approaches $j^* \to 1/R$.

These forms of steady accumulation metrics can be found in e.g. \citep{landoni2019} and in 
use in \citep{tbac2018q4}, \citep{tbac2019q2}. Note that
technically these simplified formulas only 
remain representative of the unique steady-state under a particular asymptotic limit $\gamma \to 1, r \to 0$ that 
maintains $\Phi(\gamma,{r},{f})<1$. 

If the condition $\Phi<1$ does not hold, for example for $\gamma$ sufficiently near $1$ ($g \approx 0$), debt growth will instead generally be
interest-rate driven. Moreover there need not be a
single steady-state, as the spectral radius condition $\rho(B)<1$ does not hold, so we do not have asymptotic convergence to a particular solution. There can be periodic or 
otherwise non-converging debt dynamics, and those dynamics can depend on initial values.

Finally, if there is a steady-state in the non-deficit-driven case, alternate 
formulas for steady portfolio metrics will apply, generally 
having the same form as those given above, but with $\gamma$ replaced by $\beta := 1+WAC$, where $WAC$ is the
endogenously-determined steady interest cost. We omit the details here as outside the scope of this piece.

\section{Stochastic Extension}
\label{sec:stoch}

In this section we incorporate stochastic interest rates and primary deficits to better capture real-world uncertainties, allowing for 
correlations that reflect economic cycles. The basic underlying structure of the model is the same as and builds on that developed in section~\ref{sec:det}. 

\subsection{Stochastic Model Setup}

$N_t$, now stochastic, is still determined via the budget equation:
\[
N_t = D_t + I_t + M_t,
\]

Interest payments are
\[
I_t = \sum_{j=1}^M \sum_{k=1}^j r_{t-k,j} N_{t-k,j} = 
\sum_{j=1}^M f_j \left(  \sum_{k=1}^j r_{t-k,j} N_{t-k}  \right)
\]

where $r_{t,j}$ represents the $j-$period interest rate applicable to debt issued in period $t$. Because it is random and varies from one period to the next we cannot factor it out of the inner sum.

Maturing principal remains a summation of past issuance, now stochastic:
\[
M_t = \sum_{j=1}^M N_{t-j,j} = \sum_{j=1}^M f_j N_{t-j}
\]

The following sections will analyze this system under simple but representative assumptions about the evolution of $r_{t,j}$ and $D_t$.

\begin{assumption}[Interest Rates]
\label{ass:rates}
Each maturity-specific rate follows an AR(1) process:
\[
r_{t,j} = \bar{r}_j + \phi_j (r_{t-1,j} - \bar{r}_j) + \epsilon_{t,j}, \quad \epsilon_{t,j} {\sim} \mathcal{N}(0, \sigma_j^2), \quad 0 < \phi_j < 1.
\]
\end{assumption}

Rates here are independent across maturities for simplicity, though extensions could include cross-correlations.

\begin{assumption}[Deficits]
\label{ass:deficits}
Deficits are normalized via $D_t =\tilde{D}_t \gamma^t$, and the normalized deficit process follows:
\[
\tilde{D}_t = \bar{D}_0 + \psi (\tilde{D}_{t-1} - \bar{D}_0) + \eta_t, \quad \eta_t {\sim} \mathcal{N}(0, \varsigma^2), \quad 0 < \psi < 1,
\]
\end{assumption}

This allows fluctuations around an exponentially-growing trend, with persistence $\psi$.

\begin{assumption}[Correlations]
\label{ass:corr}
The innovations $\eta_t$ and $(\epsilon_{1,t}, \dots, \epsilon_{M,t})^\top$ are correlated as $\Sigma_j = \Cov(\eta_t, \epsilon_{j,t}) = \rho \varsigma \sigma_j$, $\rho \in (-1,1)$.
\end{assumption}

A countercyclical assumption of negative correlations ($\rho < 0$) may be seen as the economically realistic one, as it
models scenarios in which high deficits (e.g., during recessions) tend to coincide
with low rates (e.g., due to monetary policy and/or safe-haven flows), dampening fiscal stress. (A more general extension of this framework could allow for a fuller correlation 
matrix; here we posit scalar parametrization for simplicity.)

The system for $N_t$ is now a stochastic recurrence equation (SRE):
\[
{N}_t = {D}_t + {I}_t+{M}_t = 
\sum_{j=1}^M f_j \left[ \left(  \sum_{k=1}^j r_{t-k,j} N_{t-k}  \right) + N_{t-j} \right] +  {D}_t 
\]

This expresses the period-$t$ new issuance amount $N_t$ in terms of new deficits and prior 
values $N_{t-k}$, as well as prior-period 
interest rates.  The rates $r_{t-k,j}$ are not constant
but realized at time of issued, and so must be retained to fully describe this recursion. One could approach this SRE by defining an augmented-state that   involves not just $N_t$ and its lags $N_s$, but prior rates $r_{s,\cdot}$ as well. Because of 
the interaction terms $r_{t-k,j} N_{t-k}$, the system thus framed is not linear as in the
deterministic baseline case; it has a bilinear form and a direct analysis would require the theory of nonlinear Markov processes.

However, under an alternative state representation we can represent the system as a linear SRE.

\subsection{Future Cashflow Representation}
\label{sec:sre}

Rather than using $N_t$ (and lags) as state variables, let outstanding debt be defined by the full collection of its \textit{future cashflows}.

\begin{assumption}[Future Cashflows]
\label{ass:cfs}
Denote by ${Y}_t := (P_t, C_t)^T \in \mathbb{R}^{2M}$ the vector of all known and fixed 
future cashflow obligations, principal ($P_t$) and coupon ($C_t$) payments, owing to oustanding debt
at period $t$. That is, the quantity $P_{t,j}  (j = 1, \dots, M)$ is the total principal amount obligated in period $t+j$, and $C_{t,j}$ the total coupon amount, due only to outstanding debt at period $t$.
\end{assumption}

Debt evolution in this framing now consists of a) rolling  both components of the outflow vector ${Y}_t$ forward, and b) adding the effect of period-$t$ new issuance $N_t$. This
can be written
\[
P_t = S P_{t-1} + N_t f
\]
\[
C_t = S C_{t-1} + N_t R_t f 
\]
with $N_t$, ${f}$ representing the period-$t$ new-issuance amount and issue allocation (respectively), as before; $S = \delta_{j,j+1}$ an
upper-diagonal shift matrix; and $R_t$ 
the matrix of coupon streams with rates struck in period $t$:
\[
R_t := 
\begin{pmatrix}
r_{t,1} & r_{t,2} & r_{t,3} & \dots & r_{t,M} \\
0 & r_{t,2} & r_{t,3} & \dots & r_{t,M} \\
0 & 0 & r_{t,3} & \dots & r_{t,M} \\
\vdots & & \ddots & & \\
0 & \dots & & 0 & r_{t,M}
\end{pmatrix}
= 
\begin{pmatrix}
1 & 1 & \dots \\
0 & 1 & \dots \\
\vdots & & \ddots 
\end{pmatrix} \cdot diag(r_t) = U \cdot diag({r}_t)
\]

Thus
\[
{Y}_t = S' {Y}_{t-1} + N_t R'_t {f}
\]
with
\[
S':=\begin{pmatrix} S & 0 \\ 0 & S \end{pmatrix}, \quad R'_t:= ( I, U \cdot diag({r}_t))^T
\]

Because the first entries of $P_{t-1}$ and $C_{t-1}$ contain precisely the amount of existing 
principal and coupon cashflows due in period $t$ (i.e. $M_t+I_t$ in our prior
formulation), the
new-issue amount $N_t$ in turn can be written linearly in the previous state,
\[
N_t = P_{t-1,1} + C_{t-1,1} + D_t =e^T {Y}_{t-1} + D_t
\]
where $e = ( e_1^T \; e_1^T )^T$.

Thus the recurrence is simply
\[
{Y}_t = B_{t} {Y}_{t-1} + d_t
\]
where $B_t$ and $d_t$ are stochastic and depend on rate/deficit data at time $t$, \textit{but not on the state} ${Y}_{t-1}$:
\[
B_t := S'+ R'_t f e^T
\]
\[
d_t := D_t R'_t f
\]
This is a \textit{linear} stochastic recurrence equation (SRE) for ${Y}_t$, with random coefficients $B_t$, and 
driven by the exogenous stochastic terms $({r}_t, D_t)$. (See e.g. (\citep{brandt1986}.) Taken together, the augmented 
series $(Y_t, {r}_t, D_t)$ is a Markov process. Of course, because the forcing term $d_t$ depends on $D_t$ which grows like $\gamma^t$, this system is unbounded
and so we will need to normalize to analyze further.

\subsection{Normalized Linear SRE and Ergodic Convergence}
\label{sec:erg}

As in the deterministic case, define normalized quantities $\tilde{P}_t = \gamma^{-t} P_t, \tilde{C}_t = \gamma^{-t} C_t, {\tilde{Y}}_t = \gamma^{-t} {Y}_t$ and so on. The
SRE of  section~\ref{sec:sre} becomes
\begin{equation}\label{normalizedSRE}
\tilde{{Y}}_t = \tilde{B}_t \tilde{{Y}}_{t-1} + \tilde{d}_t
\end{equation}
where $ \tilde{B}_t  = \gamma^{-1} B_t = \gamma^{-1}  (S'+ R'_t f e^T )$ and $\tilde{d}_t = \gamma^{-t}  D_t R'_t f = \tilde{D}_t R'_t f$. Now the forcing term $\tilde{d}_t$ is stationary, 
with finite mean and bounded moments, because it depends only on AR(1) processes: (normalized) deficits $\tilde{D}_t$ and rates $r_t$. The companion matrix $\tilde{B}_t$ is 
stochastic and depends on rates, but is independent of the state $\tilde{Y}_s$.

The normalized linear SRE~\eqref{normalizedSRE} is now in a form that allows us to establish ergodic convergence to an invariant distribution, given
a sustainability-like condition, using 
standard Foster-Lyapunov drift conditions (\citep{meyntweedie1993},\citep{hairer2011yet}).

Consider the deterministic, nonnegative expected-value dynamics $Z_t = C Z_{t-1} + b$, where $C:=E(|\tilde{B}_t|)$ and $b:=E(|\tilde{d}_t|)$. Both
$C$ and $b$ 
are constant, due to the stationarity of $\tilde{D}_t$ and $r_t$. (We emphasize that this recursion equation is merely an 
associated equation introduced for the purpose of establishing 
the required sustainability condition, and not descriptive of debt dynamics itself.) From 
the structure of $|\tilde{B}_t|$ it is clear
this system describes the dynamics of deterministic, iterated debt-issuance as described in section~\ref{sec:det}, but with interest-rates fixed at their expected absolute-values $r^*_j := E(|r_{\cdot, j}|)$.  

\begin{remark}
For completeness, even though in the deterministic case we assumed $r_j \ge 0$, here we allow for a distinction between $\bar{r}_j = E(r_j)$ and $r^*_j = E(|r_j|)$ in recognition of the fact that by driving the rate processes
with Gaussian innovations, they may become negative.  Because of this, the sustainability condition must be written in terms of $E(|r_j|)$ rather than $E(r_j)$, making it a slightly stronger constraint. This 
distinction could be removed by stipulating a lognormal or other process for $r_j$ that remains nonnegative; here we have postulated Gaussian 
innovations for ease of presentation.
\end{remark}

As we saw, such a 
system converges to a unique asymptotic steady-state in general if  a sustainability condition holds, which in this context now takes the form
\begin{equation}\label{eq:feedbackstochastic}
\Phi(\gamma, {r^*}, {f}) = \sum_{j=1}^M \left( \sum_{k=j}^M r^*_k f_k + f_j \right) \gamma^{-j}    = \sum_{j=1}^M f_j \gamma^{-j} + \sum_{j=1}^M r^*_j f_j \gamma^{-1} \frac{1 - \gamma^{-j}}{1 - \gamma^{-1}} < 1.
\end{equation}
Under this condition, the homogenous solution must decay. In the future-cashflow space of the dynamics for $Z_t$ this must equivalently mean the spectral radius is contractive, $\rho(C)<1$. By Perron-Frobenius on the 
nonnegative operator $C$, there then exists a 
nonnegative left eigenvector $v$ of $C$ with $v^T C = \rho(C) v^T$, $\rho(C)<1$. 

Use this $v$ to define a Lyapunov function on the original system, $V(Y) = v^T |Y|$. By construction,
\[
E( V( \tilde{Y}_{t+1} ) | Y ) \le v^T E(  |\tilde{B}_t| ) |Y|   + E(\tilde{d}_t) = v^T C |Y| + v^T b = \rho(C) v^T |Y| + b'  = \rho(C) V(Y) + b'. 
\]
where $b'$ is a constant, and $\rho(C)<1$. The Foster-Lyapunov drift criterion is thus satisfied.

\begin{remark} (Minorization).  Because the exogenous drivers $r_t, \tilde{D}_t$ include Gaussian innovations with full support on
$\mathbb{R}^k$, their conditional distribution given any past state, thus the affine term $\tilde{d}_t$, and therefore the transition 
kernel of the update map for $\tilde{Y}_t$, 
admits a density that is strictly positive on every compact set.
Therefore, every sufficiently large compact sublevel set of the Lyapunov function is a small set in the sense of Meyn–Tweedie, establishing
 the required minorization condition as well.
\end{remark}

These drift and minorization conditions establish that the system converges ergodically.

\begin{proposition}[Deficit-Driven Regime and Ergodicity]
\label{prop:stoch_sustain}
Under the sustainability condition $\Phi(\gamma, E(|r|), {f}) < 1$ (equivalently, $WAC( E(|r|), f, \gamma) < g$), the normalized debt process $\tilde{Y}_t$ 
driven by exogenous primary-deficits growing at rate $\gamma$ converges ergodically
to a unique invariant measure $\pi$.
\end{proposition}

In particular, for any bounded measurable $g$,
\[
\frac{1}{T} \sum_{t=1}^T g({\tilde{Y}}_t) \to E_\pi[g] \quad \text{a.s.}
\]
(Below, where written, $E[\cdot]$ will always denote $E_{\pi}[\cdot]$ for this invariant measure $\pi$.)

\begin{remark}
While here we have developed and illustrated the stochastic structure stipulating 
simple AR(1) Gaussian processes for rates and deficits, the analysis can be extended and equally applied
to more realistic macro environments in which these variables are, for 
example, linear functions of underlying Gaussian factors (e.g. output gap, inflation) governed by a VAR. In such settings, the normalized debt dynamics 
in the future-cashflow representation
remain a linear stochastic recurrence equation driven by a stationary exogenous process, and the ergodicity and the invariant-mean results below continue 
to hold with the same formulas, with $E(r_t)$, $E(\tilde{D}_t)$ and 
$Cov(r_t,\tilde{D}_t)$ interpreted as the corresponding moments implied by the macro-factor model.
\end{remark}

What we have observed here is that in a simple but representative class of stochastic debt dynamics, under a sustainability condition that implies deficit-driven growth,
there is a unique invariant measure to which the portfolio converges ergodically (in distribution). This in particular means that pathwise time-averages of 
any well-behaved (measurable) portfolio metrics in such a system recover the 
corresponding functionals of the steady-state distribution.

\subsection{Invariant Mean State}

Return to the normalized, linear SRE:
\[
\tilde{{Y}}_t = \tilde{B}_t \tilde{{Y}}_{t-1} + \tilde{d}_t
\]
Assuming $\Phi<1$ such that we have ergodic convergence, because $\tilde{B}_t$ is independent of $\tilde{Y}_t$, the invariant mean satisfies
\[
E(\tilde{Y}) = (I - E(\tilde{B}))^{-1} E(\tilde{d})
\]
Here and below we may (due to invariance of $\tilde{Y_t}$ and stationarity of $\tilde{D}_t$, $r_t$) suppress dependence on $t$.

This system for $E(\tilde{Y})$ is similar to the deterministic 
baseline system of section~\ref{sec:det}. By its form, $E(\tilde{B})$ is the companion matrix for the deterministic system, with rates replaced by their means $\bar{r}_j$. However, 
$E(\tilde{d})$ includes interaction between rates and deficits. Recall $\tilde{d} = \tilde{D} R' f$ where $R'$ depends on rates, having the form
\[
R' = 
 \begin{pmatrix}
  I \\
   R
   \end{pmatrix}
:=
 \begin{pmatrix}
  I \\
   U \cdot diag({r})
   \end{pmatrix}
\]
Thus 
\[
E(\tilde{d}) =
\begin{pmatrix}
\bar{D_0} I \\
U \cdot diag( \bar{D}_0 \bar{r} + \Sigma )
\end{pmatrix}
f
= \bar{D}_0 \bar{R}'_{\Sigma} f,
\]
recalling that $\Sigma_j =  \rho \varsigma \sigma_j$ is the covariance between $\tilde{D}$ and $r_j$, and writing
\[
\bar{R}'_{\Sigma} =
\begin{pmatrix}
I \\
\bar{R}_{\Sigma}
\end{pmatrix}
:= \begin{pmatrix}
I \\
U diag( \bar{r} + \Sigma/\bar{D}_0 )
\end{pmatrix}
\]

The inverse operator $(I-E(\tilde{B}))^{-1}$ can be represented directly due to 
the form of $\tilde{B}$. Recall $\tilde{B} = \gamma^{-1} B$ where $B$ has the form of a rank-one update to a (doubled) shift operator,
\[
B = S' + (R' f) e^T
\]
The sustainability condition $\Phi(\gamma,\bar{r},f)<1$ allows use of the Sherman-Morrison formula, in this context leading to
the explicit formula for the inverse,
\[
(I - E(\tilde{B}))^{-1} = T' \left(I + \frac{1}{1-\Phi} 
\bar{R}'
f  \gamma^{-1} h_2^T
  \right)
\]
where
\begin{align*}
\bar{R}' &= (I, E(R))^T = (I,U diag(\bar{r}))^T, \\
U  &\text{ is the upper-triangular matrix of 1s}, \\
h_2^T &:= (h^T,h^T), \\
h &:=(1,\gamma^{-1}, \dots, \gamma^{-(M-1)})^T, \text{and} \\
T' &:=(1-\gamma^{-1}S')^{-1}.
\end{align*}
$T'$ in turn has the block-Toeplitz form
\[
T' = \begin{pmatrix}
T & 0 \\
0 & T
\end{pmatrix},
\quad  \quad
T = 
\begin{pmatrix}
1 & \gamma^{-1} & \gamma^{-2} & \dots \\
0 & 1 & \gamma^{-1} & \dots \\
0 & 0 & 1 & \dots \\
\vdots & & & 
\end{pmatrix}
\]
(The latter identity can be established by writing $T = I + \gamma^{-1}S + \dots + \gamma^{-(M-1)}S^{M-1}$ and noting $(I-\gamma^{-1}S)T$ is a telescoping sum, with the final term vanishing because $S^M=0$.)
Thus also $h = T^T e_1$.

The expected normalized state $E(\tilde{Y})$ is then
\[
E(\tilde{Y}) = (I - E(\tilde{B}))^{-1} E(\tilde{d})
= T' \left(I + \frac{1}{1-\Phi} 
\bar{R}'
f \gamma^{-1} h_2^T 
  \right) \cdot
  \bar{D}_0 \bar{R}'_{\Sigma} f
\]

Note that $\gamma^{-1} h_2^T \bar{R}'_{\Sigma} f = \Phi( \bar{r}_{\Sigma})$, where $\bar{r}_{\Sigma} := \bar{r} + \Sigma/\bar{D}_0$, and factor it out of the second term. We find
\begin{align} 
E(\tilde{Y}) &= \bar{D}_0 T' \bar{R}'_{\Sigma} f + \bar{D}_0 \frac{ \Phi(\bar{r}_{\Sigma} )}{1-\Phi(\bar{r})} T' \bar{R}' f   \nonumber \\
     &= \frac{\bar{D}_0}{1-\Phi(\bar{r})} T'
        \left( (1-\Phi(\bar{r})) \bar{R}'_{\Sigma} + \Phi(\bar{r}_{\Sigma} ) \bar{R}' \right) f. \label{eq:ey}
\end{align}
This demonstrates the expected debt-state as a linear transformation of the issuance vector $f$. This transform is a linear
combination of the maturity/interest rate cumulative operators $\bar{R}', \bar{R}'_{\Sigma}$, multiplied by the
discounting matrix $T'$ (which discounts periodwise by $\gamma^{-1}$). It also scales up with the normalized 
expected-deficit level $\bar{D}_0$ and varies inversely with how close the average feedback $\Phi(\bar{r})$ is to $1$.

Note in particular that if $\Sigma \to 0$ (recovering the deterministic case), so that
$\bar{r}_{\Sigma} \to \bar{r}$ 
and $\bar{R}_{\Sigma} \to \bar{R}$, this expression collapses to $E(\tilde{Y}) = \frac{\bar{D}_0}{1- \Phi(\bar{r})} 
T' \bar{R}' f$, which is easily seen to be equivalent to the deterministic steady-state of 
section~\ref{sec:det} in future-cashflow representation, using $\bar{r}$ as the steady interest rate assumption. Expression~\ref{eq:ey} therefore
has the following intuitive interpretation: 

\begin{proposition}[Stochastic Invariant Mean]
The invariant mean state of SRE \eqref{normalizedSRE}, if the sustainability condition is satisfied,  is 
given by \eqref{eq:ey} and differs from
the deterministic baseline of section \ref{sec:det} (using average deficits/rates) by a correction term that depends
linearly
on the correlation between interest rates and deficits.
\end{proposition}

The invariant mean state $E(\tilde{Y})$ can be helpful in characterizing various steady portfolio metrics, analogous to those developed
in section \ref{sec:det}.

\subsection{Invariant Maturity Distribution}

The invariant maturity distribution of SRE \eqref{normalizedSRE} is obtained from the principal portion of $E(\tilde{Y})$, which are just its first $M$ elements:
\begin{align*}
E(\tilde{q}) & := ( I, 0) E(\tilde{Y}) =  \frac{\bar{D}_0}{1-\Phi(\bar{r})} ( I, 0)  T'
        \left( (1-\Phi(\bar{r})) \bar{R}'_{\Sigma} + \Phi(\bar{r}_{\Sigma} ) \bar{R}' \right) f \\
 & = \frac{\bar{D}_0 ( 1 -\Phi(\bar{r}) + \Phi(\bar{r}_{\Sigma} ))  }{1-\Phi(\bar{r})} Tf \\ 
& =\frac{\bar{D}_0}{1-\Phi(\bar{r})} \left( 1 + h^T U diag(\Sigma/\bar{D}_0) f \right) T f.
\end{align*}
where we have used $(I,0)T' = (T,0)$ and $(T,0)\bar{R}' = (T,0)\bar{R}'_{\Sigma} = I$.

Note that $(Tf)_j = \sum_{k=j}^M \gamma^{-(k-j)} f_k$. 
Comparing with \eqref{eq:qlevels}, we see that the above equation has the same form as the steady normalized debt-levels $q_j$ obtained
in the deterministic case, up to a constant (an adjustment that scales with
rate/deficit correlation $\Sigma$). Since
portfolio shares $\theta = \tilde{q} / \mathbf{1}^T \tilde{q}$ are normalized debt levels scaled to add to $1$, the constant cancels out and plays
no role in $\theta = \tilde{q} / \mathbf{1}^T \tilde{q} = Tf / \mathbf{1}^T T f$. Thus moving to a stochastic regime with correlated, mean-reverting
rates and deficits has left steady (invariant) portfolio shares unaffected versus the deterministic baseline model.

\begin{proposition}[Invariant Portfolio Shares]
The invariant portfolio shares $\theta$ for the stochastic SRE \eqref{normalizedSRE}, provided the sustainability condition $\Phi(\gamma,r^*,f)<1$ is satisfied, 
are given by the formula for deterministic steady shares in
\eqref{eq:shares_det}, which is equivalent to $\theta = Tf / \mathbf{1}^T T f$, and are otherwise independent of rates, deficits, or their correlation.
\end{proposition}

\subsection{Invariant Rollover Fraction}

It also follows from the preceding that the invariant rollover fraction $\theta_1$ is identical to that in the deterministic case, because 
the correlation-correction cancels; we write $\theta_1=\theta_{1,base}$
where $\theta_{1,base}$ is given in equation \eqref{eq:rr}:
\[
\theta_1 = \theta_{1,base} = \sum_{j=1}^M w_j \tau_j, \quad  \tau_j = (\gamma-1)/(\gamma^j-1), \quad w_j = f_j (1-\gamma^{-j}) / \sum_k f_k (1-\gamma^{-k}).
\]

\subsection{Invariant Debt Level}

The total (normalized) debt outstanding, in the invariant steady state, is the summed entries of $E(q)$;
\begin{equation}\label{eq:EQstochastic}
E(\tilde{Q}) = \mathbf{1}^T E(\tilde{q}) 
= \frac{\bar{D}_0}{1-\Phi(\bar{r})} \left(  1 +h^T U diag(\Sigma/\bar{D}_0)f \right) \mathbf{1}^T T f.
\end{equation}

Comparing with \eqref{eq:edebtlevel}   and \eqref{eq:shares_det}, and 
noting
 $(\mathbf{1}^T T)_k = \sum_{j=0}^{k-1} \gamma^{-j} = \frac{1-\gamma^{-k}}{1-\gamma^{-1}}$, we see that the steady debt level for the SRE
formulation again differs from that for the deterministic baseline $\tilde{Q}_{base}:=\frac{\bar{D}_0}{1-\Phi(\bar{r})}\mathbf{1}^T T f$ only by the correlation-linked scale
factor $ 1 + h^T U diag(\Sigma/\bar{D}_0)f$. 

\begin{align*}
E(\tilde{Q}) &= ( 1 +h^T U diag(\Sigma/\bar{D}_0) f ) \tilde{Q}_{base}  \\
 &= \left( 1 + \frac{1}{\bar{D}_0} \sum_{j=1}^M \left( \frac{ 1- \gamma^{-j} }{1-\gamma^{-1}} \right) f_j \Sigma_j \right) \tilde{Q}_{base}.
\end{align*}

When $\Sigma = 0$ this factor is just $1$ and the deterministic baseline is recovered. Recalling 
that $\Sigma_j = \rho \varsigma \sigma_j$, the effect of correlation $\rho$ is
straightforward: if $\rho<0$ ($\rho>0$) then the normalized steady debt level $E(\tilde{Q})$ is reduced (increased) from the deterministic baseline.

In particular, for the economically-realistic case $\rho <0$, in which rates tend to be lower in times of higher deficits and vice versa, this model implies 
a lower steady debt-level $E(\tilde{Q}) < \tilde{Q}_{base}$ all else equal, as intuition would suggest.

\subsection{Invariant Interest Cost and Ratio}

The expected next-period (normalized) interest cost is given by the $M+1$st element of $E(\tilde{Y})$, which can be obtained by multiplying by
$( 0^T, e_1^T)$. Noting
\[
( 0^T, e_1^T) T' = ( 0^T, e_1^T) 
\begin{pmatrix}
T & 0 \\ 0 & T \end{pmatrix} = 
(0^T e_1^T T) = (0^T h^T),
\]
we find
\begin{align}
E(\tilde{I}) & =
\frac{\bar{D}_0}{1-\Phi(\bar{r})} (0,h^T) 
\left(
(1-\Phi(\bar{r})) \bar{R}'_{\Sigma} + \Phi(\bar{r}_{\Sigma} ) \bar{R}'
\right) f \\
& = \frac{\bar{D}_0}{1-\Phi(\bar{r})}  h^T U
\left( (1-\Phi(\bar{r})) diag(\bar{r}+\Sigma/\bar{D}_0) + \Phi(\bar{r}_{\Sigma} )  diag(\bar{r})
\right) f \\
& = \frac{\bar{D}_0}{1-\Phi(\bar{r})}  h^T U
\left(
diag(\bar{r}) + (1-\Phi(\bar{r})) diag(\Sigma/\bar{D}_0) + (\Phi(\bar{r}_{\Sigma})-\Phi(\bar{r})) diag(\bar{r})
\right) f \\
& = \frac{\bar{D}_0}{1-\Phi(\bar{r})} h^T U
\left(
diag(\bar{r}) + (1-\Phi(\bar{r})) diag(\Sigma/\bar{D}_0) + (\Phi(\Sigma/\bar{D}_0)-\Phi(0))diag(\bar{r})
\right) f.
\label{eq:EIstochastic}
\end{align}

The contribution here involving $h^T U diag(\bar{r}) f$, after normalizing by portfolio size, equates to the 
steady interest cost for deterministic issuance with interest rates $\bar{r}$ (compare
equation \eqref{eq:wac}). The other terms are correlation corrections that depend linearly on $\Sigma$. We can thus write
\begin{align*}
E(\tilde{I}) &= I_{base}(\bar{D}_0, \bar{r}) + \frac{1}{1-\Phi(\bar{r})} h^T U \left(
(1-\Phi(\bar{r})) diag(\Sigma) + (\Phi(\Sigma)-\Phi(0))diag(\bar{r})
\right)  f \\
& = I_{base}(\bar{D}_0, \bar{r}) + \frac{1}{1-\Phi(\bar{r})}  \sum_{j=1}^M   (\frac{1-\gamma^{-j}}{1-\gamma^{-1}}) f_j 
\left( 
(1-\Phi(\bar{r})) \Sigma_j  +
(\Phi(\Sigma) - \Phi(0)) \bar{r}_j 
\right)
\end{align*}

To gauge the interest cost ratio we can divide by the steady (normalized) portfolio size $E(\tilde{Q})$ (see previous section). This gives
\begin{align*}
\frac{E(\tilde{I})}{E(\tilde{Q})}& = \frac{
I_{base}+  \frac{1}{1-\Phi(\bar{r})} h^T U (  (1-\Phi(\bar{r})) diag(\Sigma) + (\Phi(\Sigma)-\Phi(0))diag(\bar{r}) ) f }
{
( 1 + h^T U diag(\Sigma/\bar{D}_0) f  ) \tilde{Q}_{base}
} \\
&= \frac{1}{1 +h^T U diag(\Sigma/\bar{D}_0) f  } \left( \frac{I}{Q} \right)_{base} + \\
& \frac{  h^T U
\left(
(1-\Phi(\bar{r})) diag(\Sigma) + (\Phi(\Sigma)-\Phi(0))diag(\bar{r})
\right) f
}{ ( 1-\Phi(\bar{r})) (1 +h^T U diag(\Sigma/\bar{D}_0) f  ) \tilde{Q}_{base} }
\end{align*}
Unlike in the deterministic case this is a nonlinear expression in issuance policy $f$, and 
its dependence on covariance $\Sigma$ may be complex and depends on the precise structure of $\Sigma$ and $\bar{r}$. 
But it is easy to see that 
as $\Sigma \to 0$ it reduces to the baseline expression $\tilde{I}_{base}/\tilde{Q}_{base} = WAC_{det}$, whose formula given in \eqref{eq:wac}.

The deviation from the baseline is  evidently not a simple linear function of $\Sigma$. 
[TODO: analyze how this depends on $\Sigma$/$rho$ better]

\subsection{Implications for Issuance Policy Optimization}
\label{sec:opt}

The ratio $E(\tilde{I})/E(\tilde{Q})$, which was used in the deterministic case, is arguably not the best indicator of interest cost in this stochastic 
context. In practice (e.g. numerical models) it is probably more 
common to track the ratio $I_t/Q_t$, thus ideally one would calculate $E(\tilde{I}/\tilde{Q}) = E(I/Q)$ to gauge the invariant average interest cost associated
with numerical metrics. A 
closed-form formula for $E(I/Q)$ in this framework
may be unwieldly here; it evidently differs from
$E(\tilde{I})/E(\tilde{Q})$ by a Jensen effect due to convexity of $1/x$, and by the covariance between $I$ and $Q$ (presumably positive), both effects of which
ought to lead to
$E(I/Q)>E(\tilde{I})/E(\tilde{Q})$. The consequence is that the interest cost ratio developed above is likely to be an underestimate of interest-cost ratios 
evaluated empriically and reported in a 
stochastic debt-simulation setting. 

However, for modest values of $\rho$ it may be possible to treat this effect as negligible, or use a first-order expansion in $\rho$, to characterize the invariant interest cost ratio. This distinction also might have little effect
on conclusions drawn about the \textit{relative} cost impact of nearby issuance strategies $f$.

One might also target either $\tilde{I}$ by itself (the steady 
normalized interest cost) or $\tilde{Q}$ as the target metrics to gauge the cost impact of fiscal policy. While in 
the deterministic case there was no mathematical difference between targeting $WAC$ and $\tilde{Q}$ for optimization purposes, in the stochastic
case these can differ.

These considerations suggest one of the following approaches to 
incorporating the stochastic, correlated framework into the earlier optimization problem. For the 
metrics:

Invariant (normalized) interest cost:
\begin{align*} 
E(\tilde{I}) (f) =& 
\tilde{I}_{base}(\bar{D}_0, \bar{r}; f) + \\
 & \frac{1}{1-\Phi(\bar{r})} h^T U \left(
(1-\Phi(\bar{r})) diag(\Sigma) + (\Phi(\Sigma)-\Phi(0))diag(\bar{r})
\right)  f
\end{align*}

Invariant (normalized) debt level:  
\[
E(\tilde{Q})(f) =   ( 1 +h^T U diag(\Sigma/\bar{D}_0)f ) \tilde{Q}_{base} (\bar{D}_0, \bar{r}; f) 
\]

Invariant interest-cost ratio:  
\[
E(\tilde{I}) (f)  /   E(\tilde{Q})(f)
\]

Steady rollover ratio:
\[
 \theta_1 = \theta_1(f) = \theta_{1,base}(f),
\]

and given a rollover-ratio tolerance $R \in (0,1)$, 
\[
\text{minimize } E(\tilde{Q})(f) \text{ subject to } \theta_1(f) \le R, \quad 0 \le f \le 1, \quad \mathbf{1}^T f = 1.
\]
or 
\[
\text{minimize } E(\tilde{I})(f) \text{ subject to } \theta_1(f) \le R, \quad 0 \le f \le 1, \quad \mathbf{1}^T f = 1.
\]
or
\[
\text{minimize } E(\tilde{I})(f)/E(\tilde{Q}(f) \text{ subject to } \theta_1(f) \le R, \quad 0 \le f \le 1, \quad \mathbf{1}^T f = 1.
\]

Above, while $\tilde{I}_{base}$ and $\tilde{Q}_{base}$ are linear in $f$, their stochastic counterparts are not. So (unlike the deterministic
baseline) these objection functions do not in general lead to a linear programming problem.
However, they are nevertheless amenable to fast solutions using sequential quadratic programming or similar methods.

Also note that for realistically-modest choices of $\Sigma$, the relative deviations of these metrics from their deterministic counterparts $\tilde{I}_{base}$ and $\tilde{Q}_{base}$ are small. This suggests that 
scenario-analyses of or optimization solutions for $f$ will be well-approximated by and only small adjustments to the deterministic baseline metrics.

\subsection{Second-Moment And Other Metrics}

While this piece focuses on first‑moment objects - such as the invariant mean of the future‑cashflow state and the resulting expressions for expected debt levels and interest cost - in a 
Gaussian context the 
same linear structure also delivers a well‑defined characterization of second moments. In particular, the centered state $\tilde{Z}_t := \tilde{Y}_t - E( Y)$ satisfies a linear stochastic recurrence whose invariant covariance
$C = E(\tilde{Z}_t \tilde{Z}_t^T)$ solves a standard Lyapunov equation of the form
\[
vec(C) = \left( I - E[ \tilde{B}_t \otimes \tilde{B}_t^T ] \right)^{-1} vec(G_{\xi})
\]
where for $\xi_t:=\tilde{d}_t - E(d) + (\tilde{B}_t - E(B)) E(Y)$, $G_{\xi} = E(\xi_t \xi_t^T)$ is the covariance of the one‑step innovation. All components of this operator equation are 
explicit functions of the means and covariances of the exogenous macro processes, and the system is numerically tractable for realistic maturity dimensions. 

Similarly, tail-risk measures of the invariant distribution such as CDaR/CFaR (conditional debt-/financing-at-risk) could reduce to closed-form formulas in a Gaussian context. 

A fuller
development 
of these second‑moment and tail‑risk implications within this model framework is left for future work.

\section{Numerical Examples}
\label{sec:numerics}

In this section we present illustrative examples of expected invariant metrics and stochastic simulations under 
specific parameter-choices for our normalized-debt SRE.

\subsection{Baseline parameters and metrics}

Table~\ref{table:params1} presents the parameters for a baseline debt-dynamics simulation. For simplicity, we assume a rolling debt-issuance policy
using three tenors (to represent short/medium/long), $j = 1, 3, 10$.

\begin{table}[t]
\centering
\caption{Baseline Debt-Dynamics SRE Parameters}
\label{table:params1}
\begin{tabular}{l l l l}
\toprule
\textbf{Section} & \textbf{Item} & \textbf{Symbol} & \textbf{Value} \\
\midrule
\multirow{1}{*}{\textbf{Key Tenors}}
 & Issuance tenors & $j$   & $1, 3, 10$ \\ 
\addlinespace
\multirow{3}{*}{\textbf{Issuance Allocation}} 
  &                & $f_1$   & 0.4 \\ 
  & New-issuance allocation & $f_3$   & 0.5 \\ 
  &  & $f_{10}$ & 0.1 \\ 
\addlinespace
\multirow{5}{*}{\textbf{Interest Rates}}
  &    &    &  .02 \\ 
  & Mean interest rates  & $\bar{r}$ & .03 \\ 
  &  & & .05  \\ 
\addlinespace
  & Rate volatility  & $\sigma$ & $0.1 \times \bar{r}$ \\ 
\addlinespace
  & Rate persistence  & $\phi$ & $0.98$ \\ 
\addlinespace
\multirow{5}{*}{\textbf{Deficits}} 
  & Asymptotic deficit-growth rate & $g$ & $.08$ \\
  & Deficit-growth factor & $\gamma$ & $1.08$ \\
  \addlinespace
  & Mean normalized deficit & $\bar{D}_0$ & $1$ \\ 
\addlinespace
  & Normalized-deficit volatility  & $\varsigma$ & $0.1 \times \bar{D}_0 = 0.1$ \\ 
\addlinespace
  & Deficit persistence  & $\psi$ & $0.98$ \\ 
\addlinespace
\multirow{1}{*}{\textbf{Correlation}} 
  & Scalar rate-deficit correlation & $\rho$ & $-0.5$ \\
\bottomrule
\end{tabular}
\end{table}

 We can also immediately compute from these
parameters our key invariant metrics, using the formulas~\eqref{eq:feedbackstochastic},~\eqref{eq:EQstochastic} and~\eqref{eq:EIstochastic} 
developed above. The results are presented in Table~\ref{table:invariantmetrics1}. Note we include the calculated 
feedback-function to confirm $\Phi<1$ for ergodic convergence. (In Table~\ref{table:invariantmetrics1} we approximated this
by using $\bar{r}$
rather than $r^* := E(|r|)$ to compute $\Phi$, but the difference is negligible for the rate volatilities $\sigma_j$ assumed.)

\begin{table}[!ht]
\centering
\caption{Selected invariant metrics of baseline SRE.}
\label{table:invariantmetrics1}
\begin{tabular}{l l l l}
\toprule
\textbf{Metric} & \textbf{Symbol} & \textbf{Value} \\
\midrule
Feedback function & $\Phi(\gamma, r^*, {f}) $ & $\approx 0.9061 $   \\
Expected normalized debt-level & $E(Q)$ & $26.7871$ \\
Expected next-period interest cost & $E(I)$ & $1.06360$ \\
Invariant interest cost ratio & $E(I)/E(Q)$ & $0.0397$ \\
Invariant 1-period rollover fraction & $\theta_1 = E(q_1/Q)$ & $0.3492$ \\
\bottomrule
\end{tabular}
\end{table}

Finally, for numerical and Monte Carlo estimation we simulated the SRE-defined debt dynamics (in normalized form, i.e. equation~\eqref{normalizedSRE})
using the parameters shown in Table~\ref{table:MCparams}.

\begin{table}[!ht]
\centering
\caption{Monte Carlo simulation parameters.}
\label{table:MCparams}
\begin{tabular}{l l l}
\toprule
\textbf{Parameter} & \textbf{Value} \\
\midrule
Number of periods (simulation horizon) & $100$ \\
Number of paths (realizations) & $500$ \\
\bottomrule
\end{tabular}
\end{table}

\subsection{Single-path examples}

Figure~\ref{fig:sim_1path} depicts $Q_t$ and $I_t$ (on dual-axis) for a single realization of the normalized SRE 
to a horizon of $t=100$, along with
the known invariant means $E(Q)$ and $E(I)$ as dotted lines. This illustrates fluctuation around
the theoretical invariant distribution intrinsic to the system. The speed of convergence and size of the fluctuations of course
will depend on the stochastic parameters i.e. interest-rate/deficit volatility and persistence.
\noindent
\begin{figure}[H]
\centering
\includegraphics[scale=.5]{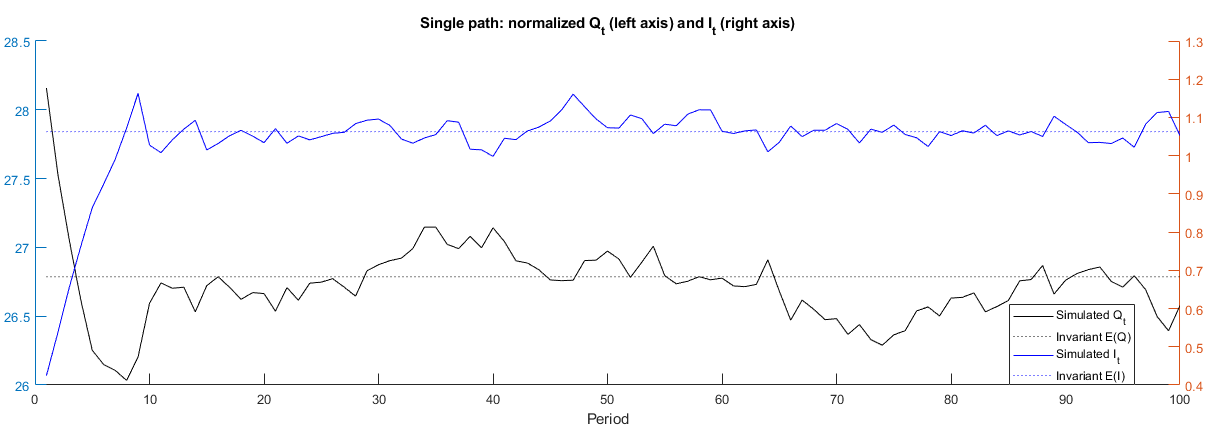}
\caption{\footnotesize Illustrative
single-path realization of $Q_t$ and $I_t$ (left- and right-axis, respectively) to $t=100$ of the baseline 
normalized SRE, along with their means $E(Q)$ and $E(I)$ (dotted lines).}
\label{fig:sim_1path}
\end{figure}

Figure~\ref{fig:sim_1theta1} demonstrates a single realization of the one-period rollover ratio $(Q_1/Q)_t$, showing that it very rapidly
settles near its theoretical long-term value.
\noindent
\begin{figure}[H]
\centering
\includegraphics[scale=.65]{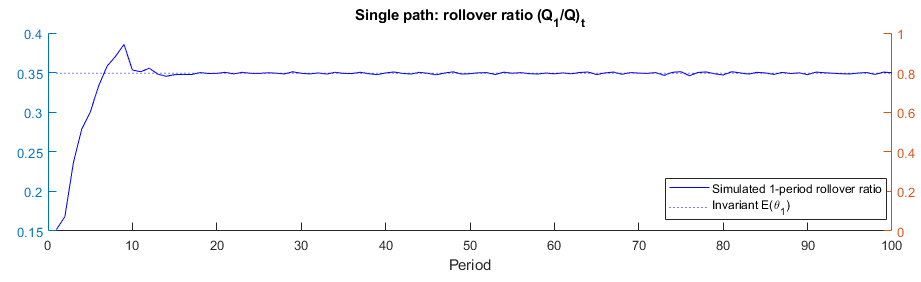}
\caption{\footnotesize Single-path realization of the one-period rollover fraction $(Q_1/Q)_t$ along with its invariant mean $\theta_1$ (dotted line).}
\label{fig:sim_1theta1}
\end{figure}

\subsection{Time averages and ergodicity}

Figure~\ref{fig:timeAvgs} depicts the ergodicity by plotting the time-averages $\sum_1^T Q_t / T$ and $\sum_1^T I_t / T$ against 
$T$ from various randomly-chosen initial conditions. Convergence to the analytical $E(Q)$ and $E(I)$ is apparent.
\noindent
\begin{figure}[H]
\centering
\includegraphics[scale=.6]{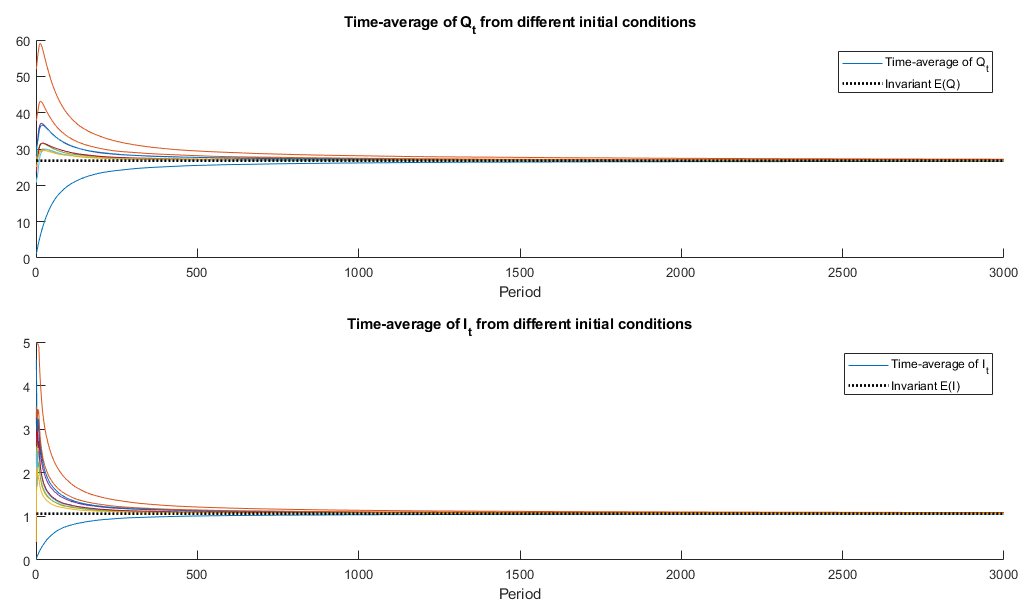}
\caption{\footnotesize Time-averages $\sum Q_t / T$ and $\sum I_t / T$ vs. $T$ 
of $Q_t$ and $I_t$ respectively from various randomly-chosen initial conditions, along with the ergodic
means $E(Q)$ and $E(I)$ (dotted lines).}
\label{fig:timeAvgs}
\end{figure}

\subsection{Distribution of simulated metrics}

Figure~\ref{fig:sim_fanQ}, ~\ref{fig:sim_fanI} and~\ref{fig:sim_fanTheta} depict 
fan-charts of the empirical distributions of $Q_t$, $I_t$ and the one-period rollover ratio $(Q_1/Q)_t$ (respectively)
based on 15th and 85th percentiles, along with their medians, across the $500$ realizations. The theoretical 
expectations $E(Q)$, $E(I)$ and $\theta_1 = E(Q_1/Q)$ are again shown as dotted-lines. (In the latter case the distribution of 
the rollover ratio apparently rapidly becomes
very tightly bound to $\theta_1$.)
\noindent
\begin{figure}[H]
\centering
\includegraphics[scale=.5]{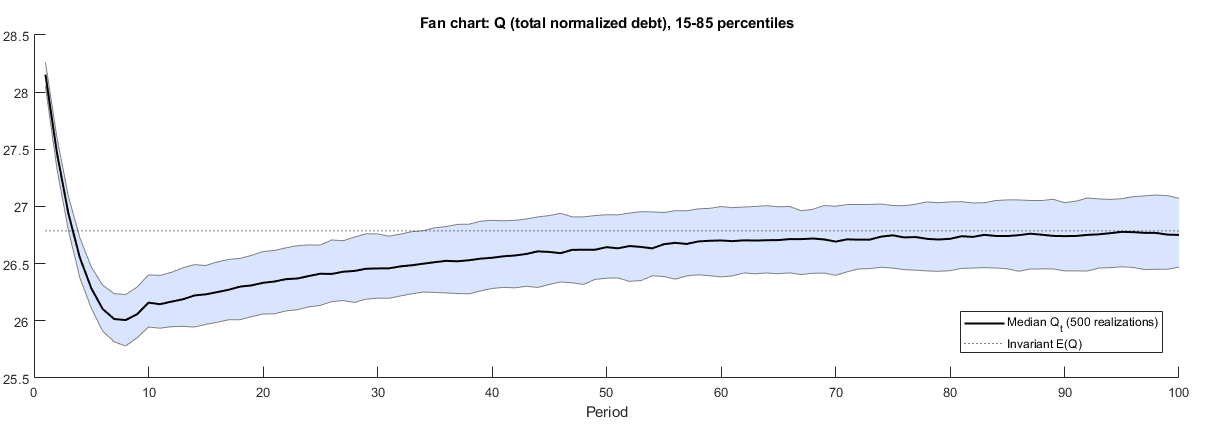}
\caption{\footnotesize Fan-chart of 15th to 85th percentile of realized $Q_t$ paths, along with median (solid line) and theoretical $E(Q)$ (dotted line).}
\label{fig:sim_fanQ}
\end{figure}

\noindent
\begin{figure}[H]
\centering
\includegraphics[scale=.5]{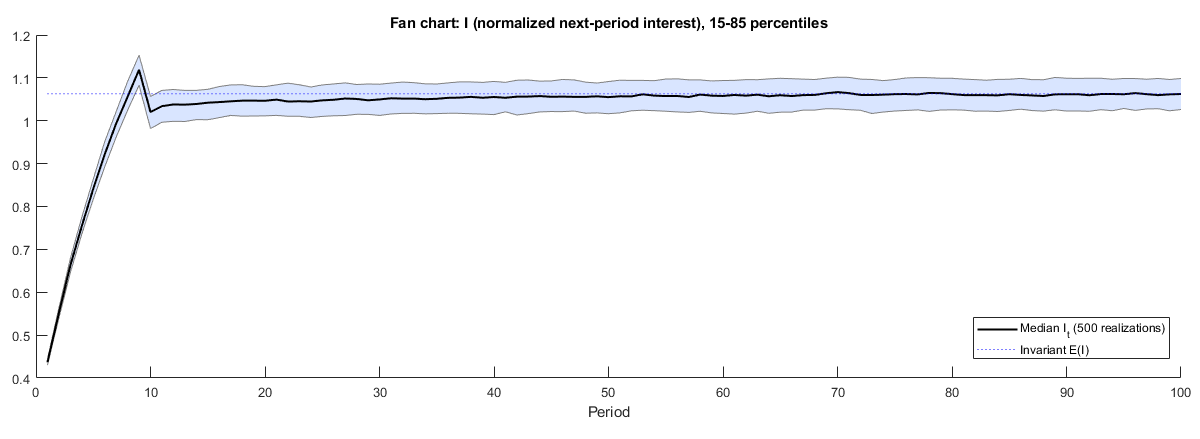}
\caption{\footnotesize Fan-chart of 15th to 85th percentile of realized $I_t$ paths, along with median (solid line) and theoretical $E(I)$ (dotted line).}
\label{fig:sim_fanI}
\end{figure}

\noindent
\begin{figure}[H]
\centering
\includegraphics[scale=.65]{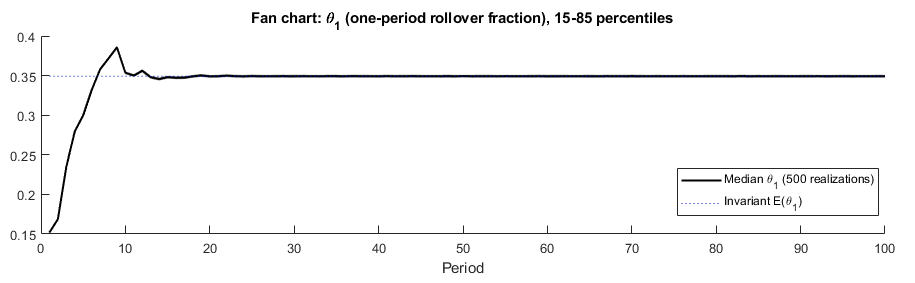}
\caption{\footnotesize Fan-chart of 15th to 85th percentile of realized $(Q_1/Q)_t$ paths, along with median (solid line) and theoretical $\theta_1 
= E(Q_1/Q)$ (dotted line).}
\label{fig:sim_fanTheta}
\end{figure}

It is seen that the empirically-estimated metrics from Monte Carlo simulation bracket the true invariant expectations, as expected.

\subsection{Cost vs. risk}

Figure~\ref{fig:costVrisk} presents a scatterplot of an ergodic-mean cost proxy $E(I)/E(Q)$ vs. risk proxy $\theta_1$ (the one-period
rollover fraction), as well as 
Monte Carlo approximations to these from a $N=500$ empirical simulation, for various representative choices of issuance 
allocations $f = (f_1, f_3, f_{10})^T$. These include the short/belly/long-tiled allocations depicted previously, as well as the baseline
allocation $f = (0.4,0.5,0.1)^T$ of this section.

One can see the close agreement between the results of forward-simulation and the analytical ergodic-mean formulas.
\noindent
\begin{figure}[H]
\centering
\includegraphics[scale=.65]{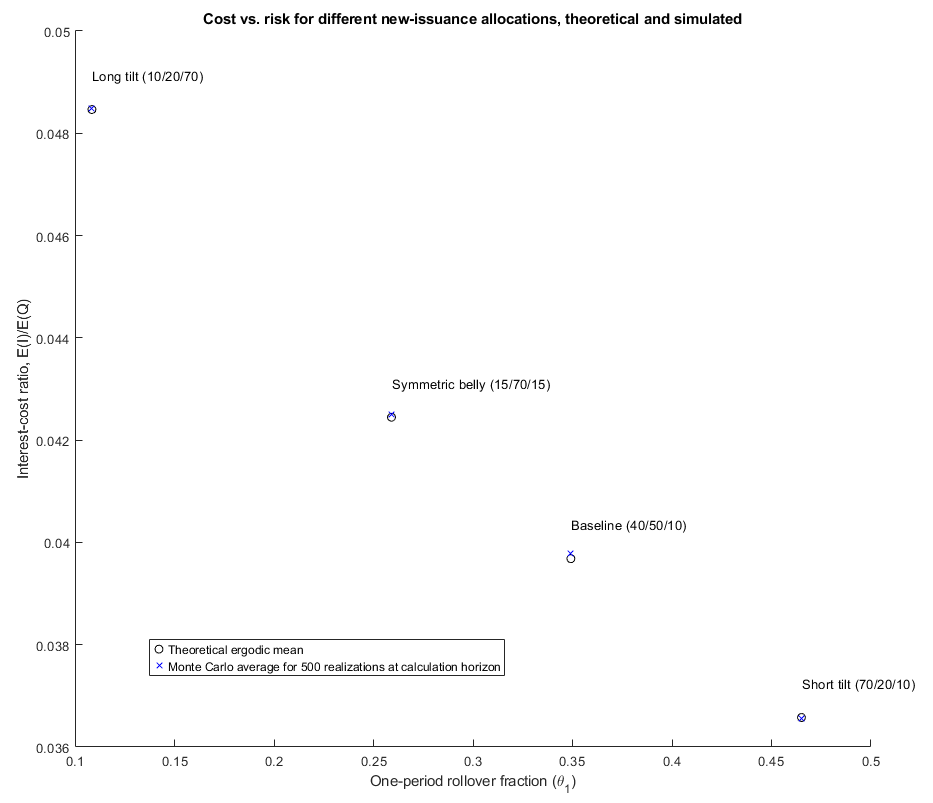}
\caption{\footnotesize Theoretical and simulated cost- vs. risk-proxy scatterplot for various new-issuance allocations. Allocations shown indicate the 
percentage of new-issuance at the 1-, 3-, and 10-year tenors respectively. Here cost is 
proxied by $E(I)/E(Q)$ (expected normalized next-period interest vs. normalized debt-level), and risk is proxied by 
$E(Q_1/Q) = \theta_1$, the expected one-period rollover fraction. Monte Carlo-approximated averages 
for $N=500$ realizations at the $t=100$ calculation horizon
show close agreement with theoretical ergodic means computable a priori via derived analytical formulae.}
\label{fig:costVrisk}
\end{figure}

\subsection{Correlation effect}

Figure~\ref{fig:costVrho} illustrates how the interest-cost ratio $E(I)/E(Q)$ varies with the rate-deficit correlation 
assumption $\rho$. One can observe the expected result that a countercyclical policy-rate assumption (i.e. $\rho<0$) leads to a lower
long-term interest-cost ratio all else equal, due to periods of higher deficits tending to coincide with those of lower interest-rates.
\noindent
\begin{figure}[H]
\centering
\includegraphics[scale=.65]{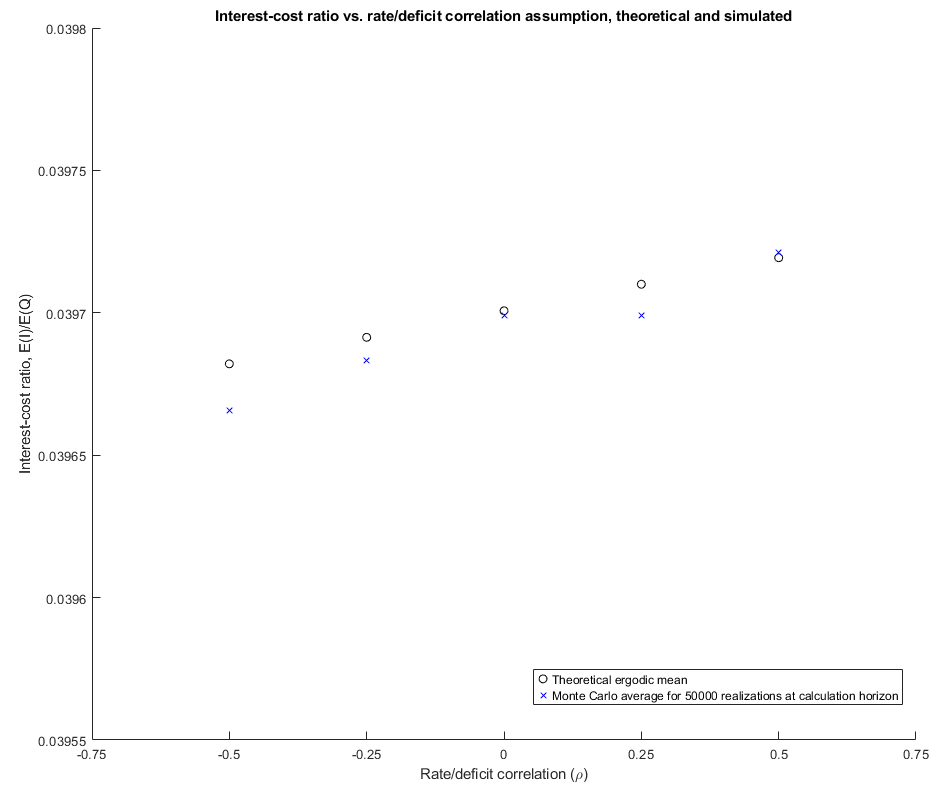}
\caption{\footnotesize Dependence of the interest-cost ratio $E(I)/E(Q)$ on rate-deficit correlation 
assumption $\rho$. Theoretical and Monte Carlo-approximated results (for $N=50000$) are shown.}
\label{fig:costVrho}
\end{figure}

Note that because the effect was so slight we have increased the number of realizations ($N=50000$) to better demonstrate 
simulation approximations approaching
the analytical formula. For completeness, Table~\ref{table:wacVrho} reports 
the dependence of $E(I)/E(Q)$ on $\rho$ for the baseline parameters and new-issuance allocation $f = (0.4, 0.5, 0.1)^T$ of
this section. For the baseline parameters of this section the
effect on expected invariant interest-cost ratio is about $0.04$ basis-points (i.e. $0.0004\%$) per $10\%$ increase in $\rho$.
\noindent
\begin{table}[H]
\centering
\caption{Dependence of interest-cost ratio ($E(I)/E(Q)$) on rate-deficit correlation assumption ($\rho$) for baseline
model and new-issuance allocation, showing the expected directional dependence (albeit slight) that reducing $\rho$ decreases interest-cost. The 
strength of this relationship could vary with other assumptions of the model such
as assumed new-issuance allocation, deficit growth trend, and the slope of the mean yield-curve.}
\label{table:wacVrho}
\begin{tabular}{c c}
\toprule
\textbf{Rate-deficit correlation, $\rho$ } & \textbf{Invariant interest-cost ratio, $E(I)/E(Q)$} \\
\midrule
$-50\%$ & $3.9682\%$   \\
$-25\%$ & $3.9691\%$   \\
$-0\%$ & $3.9701\%$   \\
$25\%$ & $3.9710\%$   \\
$50\%$ & $3.9719\%$   \\
\bottomrule
\end{tabular}
\end{table}

\subsection{Optimization calculations}

In this section we show results of running the optimization program described in section~\ref{sec:opt} on our baseline example. For
this exercise we used the simple interest-cost ratio $J(f):= E(I) / E(Q)$ as objective function. (Results of using $J(f)=E(I)$ or $J(f)=E(Q)$ were not
materially different.) 

To generate a frontier, a maximum is applied to the expected one-period rollover ratio as 
risk-proxy, $\theta_1(f) \le R$. A lower-bound of $5\%$ was enforced
on each of $f_1,f_3,f_{10}$ to avoid trivial corner solutions. The optimal allocations $f^*(R)$, admissible by this critierion, that minimize
$J(f)$ for 
$R \in (0,0.5]$ were computed and are represented in Figure~\ref{fig:opt}.
\noindent
\begin{figure}[H]
\centering
\includegraphics[scale=.8]{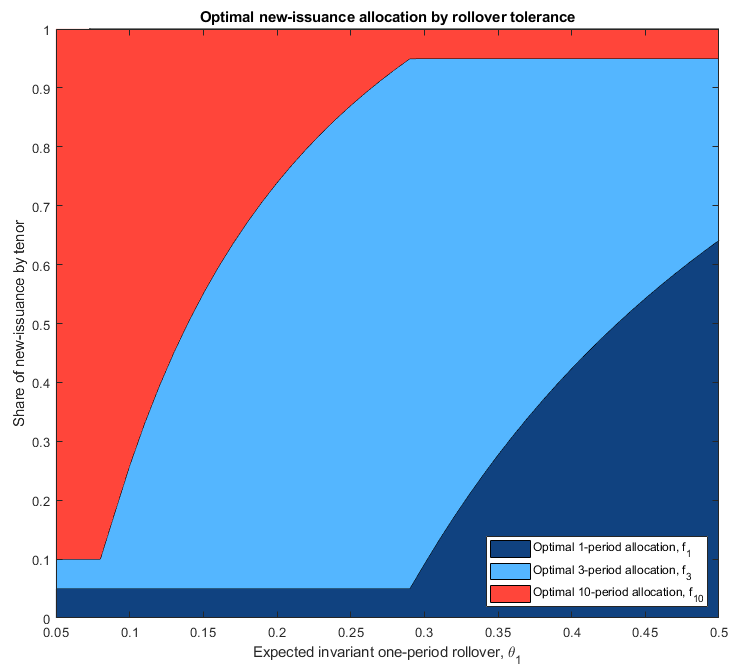}
\caption{\footnotesize Optimal new-issuance allocations $f$ of baseline model, by maximum constraint on $\theta_1$ and maintaining
$f \ge 0.05$.}
\label{fig:opt}
\end{figure}

The behavior of this optimization exercise is seen to 
follow a `waterfall' pattern familiar from the deterministic baseline model. Generally, absent any constraint (i.e. if $R$ is large), 
optimal allocation would favor shorter borrowing (large $f_1$) to minimize cost, due to the upward-sloping expected 
yields $\bar{r}$. As risk tolerance $R$ is reduced, the appetite
for $f_1$ is reduced, because it leads to too-high rollover $\theta_1$, until it reaches its minimum ($5\%$ in this case). 
Allocation toward the `belly' $f_3$ becomes favored as a cost-efficient balance. Finally, as $R$ is reduced further, $3$-period borrowing 
starts to also be inefficient for satisfying the risk tolerance, and $10$-period borrowing becomes necessary.

Finally, we explored the effect of varying rate-deficit correlation $\rho$ on optimization calculations. To illustrate this potential
dependence, Figure~\ref{fig:optByRho} shows
the optimal-calibrated one-period share $f_1$ by $\rho$ for an example case, where the constraint on one-period rollover is $R=0.3$. 

Results for all values of $R$ simulated on optimal $f_j$ (where nontrivial i.e. away from a corner solution) were similar. Correlation in this model may 
have
a nonzero effect on optimal allocation, but it is evidently not material to issuance policy
for the parameter range governing the baseline example considered in this section.
\noindent
\begin{figure}[H]
\centering
\includegraphics[scale=.6]{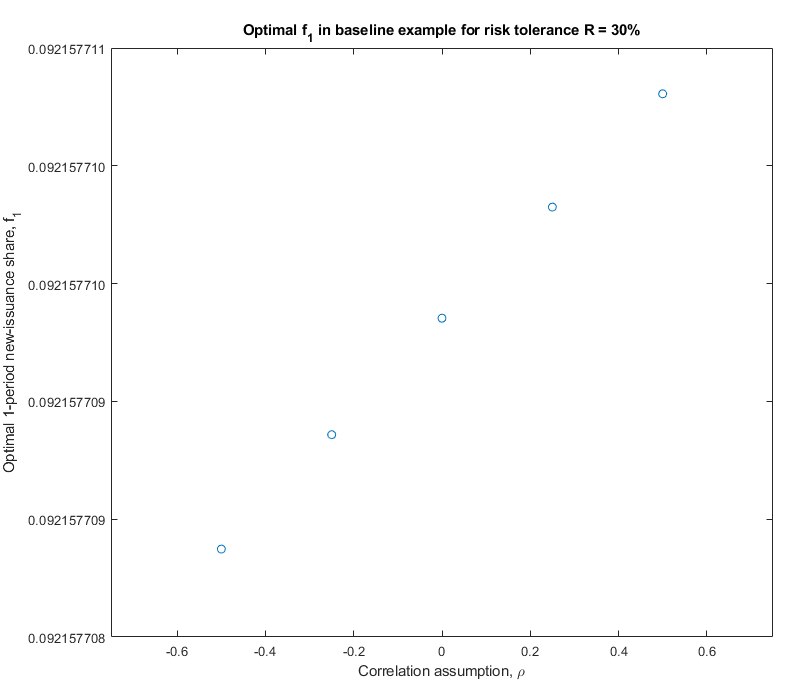}
\caption{\footnotesize Optimal one-period fraction $f_1$ in baseline model with rollover constraint $R=0.3$, while varying rate-deficit
correlation assumption $\rho$. An effect of this magnitude 
may be influenced by numerical tolerances of the optimization, and is evidently not material for issuance policy.}
\label{fig:optByRho}
\end{figure}

\section{Conclusions}
\label{sec:concl}

This piece advances the understanding of non-defaulting, regular-issuance sovereign 
debt dynamics by integrating a fully granular maturity structure with stochastic, correlated interest-rate
and deficit drivers in a discrete-time,
future-cashflow-representation recursion framework, 
and establishing a sustainability-style condition (effectively, a particularized extension of $r<g$) for ergodic convergence. 

The future-cashflow representation used here brings clear analytical benefits: in this
state-space, debt
rollover is a shift operator; discounting (i.e. accounting/normalizing for deficit growth) is a diagonal operator; new issuance is a rank-one operator. In these coordinates
the system is a linear SRE, allowing for certain analytical results.

The ergodicity result has implications for a certain class of 
Monte Carlo debt simulation models driven by stochastic, correlated fiscal and interest rate paths.
When such models conform to the conditions described in this work, their resulting calculations should approximately
recover the invariant metrics of the underlying system, provided the simulation horizon
is sufficiently large for `burn-in' to wash away the effect of initial conditions. Cognizance of this property can help 
practitioners to predict and interpret the output of such
models, as well as how or whether they are likely to respond to changes in the assumed drivers such as interest rates, deficits, and
their correlations.

 In principle, the linear SRE structure, under Gaussian forcing, may permit analytic computation or approximation
 of higher moments such as variances, tail-risk measures, or impulse responses; the focus 
of this piece has been on means and steady ratios. Future work can expand on the simple, correlation-adjusted metrics listed here.

A potential avenue of improvement is to incorporate 
floating-rate debt (e.g. FRNs and TIPS, in the US context)
into the model, for which the future-cashflow state representation may prove beneficial.

Another natural extension of this framework is to more explicitly allow rates and deficits to be determined jointly based on underlying 
factors, as in an explicit macroeconomic model. Provided such drivers remain Gaussian (or similarly well-behaved) and the relationships linear, the 
core developments and ergodicity observations described above should remain robust to such a structure.

Although the main observation here is ergodic convergence, which depended only on the feedback function, the additional
assumptions about rate/deficit behavior could
obviously affect how long the is `burn-in' period and how volatile the system remains when close to equilibrium. Thus the speed of convergence 
and the variance of paths around their ergodic means, and 
the relationship of these behaviors to the volatility and persistence of the interest-rate
and deficit driving factors, is open for further investigation.

More challenging potential extensions to this model framework
include allowing for dynamic issuance policy, and to allow issuance policy to affect interest rates 
endogenously via the well-known
supply effect. Either would disallow much of the linear machinery and analysis developed above. For 
modest supply-feedback effects, it is possible the endogenity result could still be 
demonstrated under an appropriate feedback criterion, and that first-order perturbative approaches could be used, but closed-form solutions
are unlikely.

The granular sustainability condition developed here, in its own right, may prove 
a useful analytical tool in cross-country analysis and comparison of the fiscal and debt-management policy of sovereigns. 

Ultimately, the linear SRE approach described above can help serve as a unifying backbone for modeling and practical
debt management, one that preserves analytical 
clarity while accommodating the richer economic structure required for policy‑relevant debt analysis.




\begin{thebibliography}{99}

\bibitem[Belton(2018)]{belton2018}
Belton, T, Dawsey, K., Greenlaw, D., Li, H., Ramaswamy, S. \& Sack, B.
\newblock Optimizing the maturity structure of U.S. Treasury debt: A model-based framework.
\newblock Brookings paper, available at \url{https://www.brookings.edu/articles/optimizing-the-maturity-structure-of-u-s-treasury-debt/}.



\bibitem[Blanchard(2019)]{blanchard2019}
Blanchard, O. (2019).
\newblock Public debt and low interest rates.
\newblock \emph{American Economic Review Papers and Proceedings}, 109, 1--24.

\bibitem[Bolder(2003)]{bolder2003}
Bolder, D. (2003)
\newblock A Stochastic Simulation Framework for the Government of Canada's Debt Strategy.
\newblock \emph{Bank of Canada Working Paper}, 2003.

\bibitem[Bolder(2011)]{bolder2011}
Bolder, D. and Deeley, S. (2011)
\newblock The Canadian Debt-Strategy Model: An Overview of the Principal Elements.
\newblock \emph{Bank of Canada Working Paper}, 2011.



\bibitem[Bougerol and Picard(1992)]{bougerol1993products}
Bougerol, P. and Picard, N. (1992).
\newblock Strict stationarity of generalized autoregressive processes.
\newblock {\em The Annals of Probability}, pages 1714--1730.

\bibitem[Brandt(1986)]{brandt1986}
Brandt, A. (1986).
\newblock The stochastic stability of autoregressive processes with random coefficients.
\newblock \emph{Journal of Multivariate Analysis}, 19, 1--24.

\bibitem[Cameron(2018)]{cameron2018}
Cameron, C. (2018).
\newblock Visualizing Treasury Issuance Strategy. White paper, available at \url{https://ssrn.com/abstract=3120036}.


\bibitem[Consiglio(2012)]{consiglio2012}
Consigliio, A. and Staino, A. (2012).
\newblock A Stochastic Programming Model for the Optimal Issuance of Government Bonds.
\newblock \emph{Annals of Operations Research}, March 2012.


\bibitem[Cottarelli and Jaramillo(2012)]{cottarelli2012fiscal} Cottarelli, C., \& Jaramillo, L. (2012). 
\newblock Walking hand in hand: Fiscal policy and growth in advanced economies. 
\newblock IMF Working Paper WP/12/137, International Monetary Fund, May 2012.

\bibitem[Escolano(2010)]{escolano2010}
Escolano, A. (2010).
\newblock A practical guide to public debt dynamics, fiscal sustainability, and cyclical adjustment of budgetary aggregates.
\newblock IMF Technical Paper.



\bibitem[Hairer and Mattingly(2011)]{hairer2011yet} Hairer, M., \& Mattingly, J.~C. (2011). 
\newblock Yet another look at Harris' ergodic theorem for Markov chains.
\newblock In \emph{Probability and Mathematical Statistics}, 63, 109--117.


\bibitem[Horn and Johnson(2012)]{hornjohnson2012}
Horn, R.~A., \& Johnson, C.~R. (2012).
\newblock \emph{Matrix Analysis} (2nd ed.).
\newblock Cambridge University Press.




\bibitem[Landoni et al(2019)]{landoni2019}
Landoni, M., Smith, W.T. \& Cameron, C. (2019).
\newblock Linking policy to outcomes: a simple framework for debt maturity management. White paper, available 
at \url{https://ssrn.com/abstract=3300347}.

\bibitem[Meyn and Tweedie(1993)]{meyntweedie1993}
Meyn, S.P. and Tweedie, R.L. (1993).
\newblock \emph{Markov Chains and Stochastic Stability}.
\newblock Springer-Verlag.


\bibitem[TBAC(2018)]{tbac2018q4}
Treasury Borrowing Advisory Committee.
\newblock \emph{Discussion Charts by Calendar Year, 4th Quarter
2018.} Presentation to US Treasury, available 
at \url{https://home.treasury.gov/system/files/276/CombinedChargesforArchives4thqtr2018.pdf}.




\bibitem[TBAC(2019)]{tbac2019q2}
Treasury Borrowing Advisory Committee.
\newblock \emph{Discussion Charts by Calendar Year, 2nd Quarter
2019.} Presentation to US Treasury, available 
at \url{https://home.treasury.gov/system/files/221/q22019CombinedChargesforArchives.pdf}.




\end{thebibliography}
\end{document}